%% file: main.tex
\def\confversion{0}
\def\notes{0}

\def\pagenumtitle{1}
\def\mystatus{1} %

\ifnum\confversion=1

    \newcommand{\cond}[1]{#1}
    
    \newcommand{\condelse}[2]{#1}
\else

    \newcommand{\cond}[1]{}
    
    \newcommand{\condelse}[2]{#2}
\fi

\condelse{
    \documentclass[twoside,leqno]{article}
    \usepackage[letterpaper]{geometry}
    \usepackage{siamproceedings}
    \makeatletter
    \providecommand{\qedhere}{%
      \ifmmode\else
        \unskip\nobreak\hfill\proofbox\NoEndMark
      \fi}
    \makeatother
}{
    \documentclass[12pt]{article}
    
}
\usepackage{header}

\begin{document}

\title{Subgraph Counting with Local Node Differential Privacy\condelse{}{\thanks{To appear at the ACM-SIAM Symposium on Discrete Algorithms (SODA) 2027.}}}
\author{
    \normalsize Sofya Raskhodnikova\thanks{Boston University,
    \texttt{\{sofya,ads22,wagaman,zavyalov\}@bu.edu.}}
    \and \normalsize Adam Smith\footnotemark[\value{footnote}]
    \and \normalsize Connor Wagaman\footnotemark[\value{footnote}]
    \and \normalsize Anatoly Zavyalov\footnotemark[\value{footnote}]
}
\date{\normalsize October 15, 2026 \\ \todo{turn off notes!!}}

\maketitle

\cond{
    \fancyfoot[R]{\scriptsize{Copyright \textcopyright\ 2027 by SIAM\\
    Unauthorized reproduction of this article is prohibited}}
}

\begin{abstract}
    We study the problem of approximately counting subgraphs  
    under 
    \emph{local node differential privacy} (LNDP),
    a recently introduced model that enforces strong privacy guarantees in analyses of distributed social networks.
    We design a noninteractive $(\eps,0)$-LNDP algorithm that, given any public $k$-node template graph $H$, counts occurrences of $H$ in a sensitive undirected $n$-node graph with additive error $O(n^{k-\frac{1}{2}}/\eps^{k-1})$, for any constant $k$.
    We complement this algorithm with a lower bound showing that every noninteractive $(\eps,\del)$-LNDP algorithm for counting copies of any $k$-node template graph with maximum degree greater than $1$ must have error $\Omega(n^{k - \frac{1}{2}})$.
    Thus, our algorithm achieves optimal error for counting constant-size subgraphs, up to $\poly(1/\eps)$ factors.

    Our results substantially expand the types of network statistics known to be approximable under LNDP: prior work considered only functions of the degree distribution.
    In addition, our lower bound yields the first known $\poly(n)$ gap between central and local node differential privacy when accuracy is required for all input graphs.
    Previous separations either only involved $\poly(1/\eps)$ factors or only applied to promise problems.

\end{abstract}

\condelse{}{
    {
    \hypersetup{linkcolor=black}
    \setcounter{tocdepth}{3}
    \newpage \tableofcontents
    }
    
    \newpage
}
\section{Introduction}
Graphs are a natural representation for sensitive relational data: social
networks, contact networks, and 
networks of sexual relationships
all encode information
about interactions between individuals.
Such sensitive graphs are often
distributed, with each node knowing only its own connections and no central authority holding the whole graph.
Accurate aggregate information about these datasets is highly valuable---for example, to epidemiologists seeking to understand the spread of an illness, big tech companies working to improve a messaging platform, and  governments implementing public health efforts. 
However, these benefits come with serious privacy risks:
releasing the graph structure directly can reveal sensitive relationships
between individuals~\cite{BackstromDK07}; moreover, even releasing aggregate
summaries of graph structure can leak private information about the individuals
and relationships represented in the graph~\cite{JerniganM09}.

Differential privacy (DP) \cite{DworkMNS16} provides a widely used framework for reasoning about how to balance these privacy and accuracy concerns. The \emph{local model} of differential privacy \cite{EvfimievskiGS03,KasiviswanathanLNRS11} adapts this framework to a decentralized %
setting where each party randomizes its own data before sharing it.
The noisy data can then be
aggregated by an
untrusted server without any impact on the privacy guarantees. This approach's convenience and explainability, among other factors, have led to industry deployments of the local model
\cite{ErlingssonPK14,BittauEMMRLRKTS17,DingKY17,apple-iconic-scenes23} and extensive theoretical study (see
\cw{\Cref{sec:unrelated}}).

When working with graph data, there are two natural privacy notions. \emph{Edge privacy} roughly hides the existence of a particular edge, while \emph{node privacy} roughly hides a node's entire edge set.
Both are well-motivated: edge differential privacy (which has been studied extensively in the local model under the name LEDP---see \Cref{sec:unrelated}) makes sense 
when individuals corresponds to edges, whereas node privacy is well-suited to settings where individuals map to nodes whose entire neighborhoods are sensitive, as in networks of disease transmissions or sexual relationships \cite{HodNGKTS26}.
For the latter, edge privacy still allows an analyst to approximate each node's number of partners, while node privacy hides the count entirely.

The dual considerations of node privacy and decentralization
motivate the recently introduced \emph{local node differential privacy} (LNDP) \cite{RaskhodnikovaSWZ26}.
In the \textit{noninteractive} LNDP setting, denoted \LNDP,
each node runs a local randomizer on its edge list using both public and local randomness, and then releases the output.
As formalized in \Cref{def:lndp}, an algorithm is $(\eps,\del)$-\LNDP if the randomizers' joint output distribution is $(\eps,\del)$-indistinguishable
for every pair of node-neighboring input graphs.
(Two graphs are node neighbors if one can be obtained from the other by \emph{rewiring} a single node, that is, 
changing its edge set arbitrarily.)
While \cite{RaskhodnikovaSWZ26} also present an interactive formulation, their work focuses on the noninteractive setting; this focus mirrors standard local DP, where all industry deployments and almost all theoretical work are noninteractive.
They study
the error required for estimating functions of a graph's degree distribution but leave estimating other functions of the input graph, such as subgraph counts, open.

We study subgraph counting under \LNDP.
Given a \textit{template graph} $H$ and an input graph $G$, the goal is to count the number of copies of $H$ in $G$.
Subgraph
counts are among the most fundamental graph statistics. They measure
the frequency of small structural patterns, such as triangles, paths, stars,
and cliques, and are widely used in network analysis as summaries of local
connectivity and higher-order structure. For example, triangle counts are
closely related to clustering coefficients and community structure, stars
capture degree heterogeneity, and counts of larger motifs are used to identify
functional organization in social, biological, and communication
networks~\cite{WattsStrogatz98,MiloSSSZ02,BensonGL16}.
On a theoretical level, subgraph counts characterize graph limits and the graphon statistical model (e.g., \cite{Lovasz12book}).
At the same time, subgraph counts are particularly challenging under node privacy. Changing the edges incident to one node can create or destroy many copies of a small pattern (e.g., rewiring a single node can change the number of instances of a $k$-node subgraph by $\Theta(n^{k-1})$). Moreover, in the local model no party can directly
inspect the global structure of the graph to find these copies.

\subsection{Our Results}
\label{sec:results}

We study the additive error required for counting the number of occurrences of a $k$-node template graph $H$ in an undirected $n$-node graph
under \LNDP. %
Our results are summarized in \Cref{tab:sub-results}.
Let $\#H(G)$ denote the number of noninduced,\footnote{Our algorithm can be easily adapted to instead count induced subgraphs, with the same $O_k(n^{k-\frac{1}{2}}/\eps^{k-1})$ error guarantee.}
unlabeled copies of $H$ in $G$ (\Cref{def:subgraph-count}).
We give
an \LNDP algorithm for computing $\#H(G)$ with
the following %
error guarantee.
(As is standard, the notation $O_k(\cdot)$ indicates that the suppressed constant may depend on $k$.)

\input{figures/table-anatoly-intro}

\begin{theorem}[Algorithm for subgraph counting; \Cref{thm:sub-count}, informal]
\label{thm:ub-inf}
    Fix integer $k\in \N$,\srnote{All natural numbers are integers. Do we want to say $k\geq 3$ instead?} $k$-node template graph $H$, and $\eps \in (0, 1]$. Let $\vc{H}$ denote the size of the minimum vertex cover of $H$. 
    There is an $(\eps,0)$-\LNDP algorithm $\algsubh$ that, for all $n$-node input graphs $G$, 
    counts copies of $H$ in $G$ %
    with additive error
    \[
        \Ex\bracks{\babs{\algsubh(G) - \subgraphcount(G)}}
        = O_k\paren{ \frac{n^{k-\frac12}}{\eps^{\vc{H}}} }
        = O_k\paren{ \frac{n^{k-\frac12}}{\eps^{k-1}}}.
    \]
\end{theorem}

Note that $\vc{H}$ is between 1 and $k-1$, depending on the template graph $H$, giving the last equality.
Since the subgraph count $\#H(G)$ can be as large as $\Omega(n^k)$, our algorithm %
has additive error that, for dense graphs, is vanishing on the scale of the statistic being approximated.

Our algorithms greatly expand the types of network statistics known to be approximable under LNDP: prior work \cite{RaskhodnikovaSWZ26} considered only functions of the degree distribution. The only connected subgraphs whose counts fall into that category are $k$-stars; the counts of $k$-cliques, $k$-cycles, and other interesting subgraphs require the simultaneous consideration of several nodes.

In \Cref{sec:lb}, we show this additive error is optimal for constant $k$, up to $\poly(1/\eps)$ factors.

\newcommand{\subgraphlb}{
For every $\eps \in (0, 1]$, $\del \in [0, \frac{1}{200}]$, and template graph $H$ on $k\in\N$ nodes with maximum degree greater than $1$,
    if $\algsubh$ 
    is $(\eps,\del)$-\LNDP and 
    \(
    \Ex\bracks{\babs{\algsubh(G) - \subgraphcount(G)}} \leq \alpha(n)
    \)
    for every $n$-node graph $G$, then $\alpha(n) = \Omega_k\bparen{n^{k-\frac12}}.$
}

\begin{restatable}[Error required for subgraph counting]{theorem}{SubgraphLB}
\label{thm:lb-inf}
\label{thm:lb}
    For every $\eps \in (0, 1]$, $\del \in [0, \frac{1}{200}]$,
    template graph $H$ on $k\in\N$ nodes with maximum degree greater than $1$,
    \cw{and sufficiently large $n\in\N$,}\srnote{We don't need to say that -- it is in the definition of asymptotic notation.}
    if $\cA$ 
    is $(\eps,\del)$-\LNDP and 
    \(
    \Ex\bracks{\babs{\cA(G) - \subgraphcount(G)}} \leq \alpha(n)
    \)
    for every $n$-node graph $G$, then $\alpha(n) = \Omega_k\bparen{n^{k-\frac12}}.$
\end{restatable}
\aznote{Can also mention: The same result for $H$ with maximum degree equal to $1$ would imply the edge-counting lower bound open in \cite{RaskhodnikovaSWZ26} -- even if $H$ has more than $2$ nodes, it is equivalent to edge counting.}
In combination, \Cref{thm:ub-inf,thm:lb-inf} pin down the dependence on $n$ in the error of \LNDP algorithms for estimating %
counts of a $k$-node template graph
to $\Theta_k(n^{k-\frac 12})$.
Interestingly, while a factor of $n$ separates subgraph counting 
under the central and local variants of edge DP,
there is only a $\sqrt{n}$ gap between the accuracies achievable under central and local node DP---see \Cref{tab:sub-results}.

Previously, the only bounds for subgraph counting in the local model were for LEDP 
\cite{EdenLRS25,SuppakitpaisarnPHH25}; generalizing the techniques used in those works only yields trivial bounds for our setting.
LEDP lower bounds are only known for restricted classes of subgraphs, namely cliques and cycles \cite{EdenLRS25,SuppakitpaisarnPHH25}. 
Our \LNDP lower bounds are tight and apply to arbitrary template graphs of maximum degree greater than 1.
Additionally, our lower bound
establishes
the first $\Omega(\sqrt{n})$ separation between local and central node DP for general graph problems.
While \cite{RaskhodnikovaSWZ26} gave local-central separations, they 
either involved only $\poly(1/\eps)$ factors or only applied to promise problems, in which the accuracy guarantee 
holds for
a restricted class of inputs graphs (with privacy holding for all graphs).
Our lower bounds 
are also %
more general than the analogous results for edge privacy and provide new insights into the relationship between the local and central models.

\subsection{Our Techniques}
\label{sec:techniques}

We now describe the techniques behind our \LNDP algorithm and 
additive-error lower bounds for
counting subgraphs.

\subsubsection{Algorithm for Subgraph Counting}

In the setting of local \emph{edge} differential privacy (LEDP), optimal noninteractive algorithms for subgraph counting work by releasing the graph's entire adjacency matrix via randomized response on each entry with parameter $\eps$, and postprocessing the result \cite{EdenLRS25}.
This technique breaks down under the constraint of node privacy since rewiring a single node can change $\Theta(n)$ edges (i.e., bits in the adjacency matrix); consequently, 
releasing the adjacency matrix with an $(\eps,\del)$-\LNDP guarantee requires running randomized response with parameter $O(\eps/\sqrt{n})$ on each bit. Such an approach yields an algorithm that counts $k$-node subgraphs with additive error $\Omega(n^k)$, which is trivial.

We devise a different approach for our algorithm.
It relies heavily on public randomness, which is used to partition the node set into $n/k$ blocks of $k$ nodes, each corresponding to a single candidate copy of the template graph $H$.
Within each block, the relevant nodes report, via randomized response with debiasing, whether the edges required by their assigned positions in $H$ are present.
Postprocessing multiplies these reports to obtain a noisy indicator for the candidate copy, and then rescales and averages the indicators over all blocks to obtain a final estimate for the subgraph count.

\cwnote{old text in green; new text in purple. We will choose one. (It is all my text---I just needed to distinguish between the two options)}

This random partitioning-based design has two key properties that are essential to our algorithm's (optimal) error guarantees.
The first feature, reduced sensitivity,
emerges because \emph{rewiring a single node 
can affect only the $k$ reports from its own block}. 
Randomized response with parameter $\eps/k$ thus suffices for node privacy, which keeps the privacy error small.
The second feature is that, after rescaling, the noisy indicators we extract from each of the blocks are unbiased, low-variance estimators of the total subgraph count.
Partitioning the nodes, as opposed to sampling blocks with replacement, makes the accuracy analysis delicate since the block indicators are no longer independent. Nevertheless, we show through a second-moment calculation that the correlations between these indicators are mild, so the sampling error is low.

Averaging the $n/k$ rescaled signals
yields an unbiased estimate with 
two sources of additive error: sampling contributes $O_k\bparen{n^{k-\frac12}}$, and privacy noise contributes $O_k\bparen{n^{k-\frac12}/\eps^k}$. In fact, for each block it suffices to obtain reports only from nodes in a vertex cover of $H$, improving the dependence from
$\eps^k$ to $\eps^{\vc{H}}$ and giving the error stated in \Cref{thm:ub-inf}.

\subsubsection{Lower Bound for Subgraph Counting}
\label{sec:intro-lb-tech}
    
Our lower bound (\Cref{thm:lb-inf}) %
is proven via a reduction from accurately answering inner product queries under differential privacy in the \emph{central model}. Given a secret dataset $\sv \in \{0, 1\}^n$, an \emph{inner product query} (or \emph{linear query}) on $\sv$ is specified by a vector $\iq \in \{0, 1\}^n$, and maps $\sv$ to the answer $p = \inn{\sv, \iq}$.
A seminal result in privacy %
\cite{DinurN03, DworkMT07} %
demonstrates that algorithms that %
answer a set of $\Theta(n)$ inner product queries $\iq^{(1)}, \ldots, \iq^{(q)}$ 
to within error roughly $\sqrt{n}$ can be used to reconstruct
an overwhelming majority
of %
$\sv$. Such algorithms, of course, are not differentially private even for modest privacy parameters.

Our reduction shows that if $\cA$ is an \LNDP algorithm with additive error $\alpha$ for $H$-counting (i.e., estimating $\#H(G)$) on $n$-node graphs, then it can be used to construct a DP algorithm $\cB$ in the \emph{central} model with additive error $O(\alpha / n^{k-1})$ for answering inner product queries. Thus, if $\alpha \leq c n^{k-1/2}$, then $\cB$ achieves error $O(\sqrt{n})$, which contradicts DP for sufficiently small constant $c > 0$.

We now describe the main ideas behind our reduction.
Algorithm $\cB$ in our reduction makes \emph{two} calls to $\cA$ on graphs that depend on $\sv$ but not on the queries. The algorithm then answers all inner product queries by postprocessing the outputs of these calls.
For each inner product query $\iq$, it 
combines these stored values to emulate the output of $\cA$ on a \emph{query graph} whose subgraph count reflects the unknown inner product value,
without explicitly running $\cA$ on each query graph.

The query graph corresponding to $\iq$ embeds $\inn{\sv, \iq}$ in its edges. 
For our reduction to work, the query graph must satisfy three requirements:
\newcommand{\Rlist}{
    \begin{enumerate}
        \item[R1:] 
        \label{itm:req-rewire}
        A single bit change in $\sv$ must correspond to a rewiring of a constant number of nodes, so that neighboring datasets in the inner product problem correspond to graphs at constant node distance. %
        \item[R2:] 
        \label{itm:req-err}
        An error of $\alpha$ in the subgraph count must translate to an error of $O(\alpha / n^{k-1})$ in the estimate for the inner product query. 
        \item[R3:] 
        \label{itm:req-calls}
        Each node whose neighborhood depends on 
        the sensitive vector $\sv$ must have one of \emph{two} possible neighborhoods in each of the $q$ query graphs, so that two runs of this node's randomizer suffice to release the relevant information about these neighborhoods.

    \end{enumerate}
}
\Rlist

\input{figures/query-graph-new}

We give the following construction of the query graph satisfying these requirements. 
Suppose w.l.o.g.\ that node 1 in $H$ is connected to nodes 2 and 3.
The query graph for an inner product query $\iq$ 
has $kn$ nodes,  %
equipartitioned
into $k$ parts $U_1, \ldots, U_k$, each corresponding to a node of $H$;
edges of $H$ correspond %
to complete bipartite graphs between the %
respective parts, except for the edges %
from
$U_1$ %
which encode the input: %
the $i$'th node of $U_1$ connects to all nodes in $U_2$ iff $x_i = 1$, and to all
nodes in
$U_3$ iff $a_i = 1$. 
Note that a change in a single bit of $\sv$ corresponds to rewiring a node in $U_1$ (thereby satisfying requirement R1), %
which can change the subgraph count in the query graph by $\Theta(n^{k-1})$.
\Cref{fig:querygraph} depicts
the query graph
when $H$ is a triangle.

While variants of the query graph have appeared before (see \cite{EdenLRS25,SuppakitpaisarnPHH25}), our construction is
novel
in several ways: it embeds an inner-product query rather than a similar but meaningfully different ``outer-product'' query, %
encodes the secret vector $\sv$ differently, and applies to arbitrary template graphs with maximum degree greater than $1$, rather than just cliques and cycles.

\input{figures/c4-query-graph}

A central %
difficulty
is in satisfying requirement R2. %
In the special case that %
$H$ is a $k$-clique, every copy of $H$ %
must use exactly one node in each part %
of the query graph, 
leading to a subgraph count of %
$n^{k-1} \inn{\sv, \iq}$. Thus, dividing an estimate for the subgraph count by %
$n^{k-1}$ yields an estimate for $\inn{\sv, \iq}$, 
with the error similarly scaling down by a factor of $n^{k-1}$.
However, the general case is more involved: for example,
when $H$ is a 4-cycle, \Cref{fig:c4-querygraph} shows that copies of $H$ %
may use multiple nodes from one part, with the subgraph count %
depending on $\|\sv\|$ and $\|\iq\|$ rather than just $\inn{\sv, \iq}$.

We resolve this difficulty by building the query graph from the concatenated vectors $\svc = \sv \circ (\vec 1 - \sv)$ and $\iqc = \iq \circ \vec 0$, which satisfy $\|\svc\| = n, \|\iqc\| = \|\iq\|$, and $\inn{\svc, \iqc} = \inn{\sv, \iq}$.
Since the weights of both $\svc$ and $\iqc$ are known, the $H$-count of $G[\svc, \iqc]$ is a function $g(\cdot)$ of just one private unknown $\inn{\sv, \iq}$. A single bit change in $\sv$ now rewires two nodes of the query graph, which still satisfies requirement R1.
As we show, $g(\cdot)$ is a quickly increasing function of $\inn{\sv, \iq}$ with a discrete derivative of 
$\frac{\partial g}{\partial \inn{\sv, \iq}} = \Omega(n^{k-1})$,
allowing us to invert it to obtain an estimate for $\inn{\sv, \iq}$.\footnote{Since $g$ is only defined on the integers, instead of the derivative we use the finite difference $\Delta g(b) = g(b+1) - g(b)$. See \Cref{sec:lb} for details.}
Then, given an estimate $d$ with error $\alpha$ for the subgraph count,
returning $g^{-1}(d)$ as our estimate for $\inn{\sv, \iq}$ gives an error of
$|g^{-1}(d) - \inn{\sv, \iq}| \lesssim \alpha / \frac{\partial g}{\partial \inn{\sv, \iq}} = O \left(\alpha / n^{k-1}\right)$.

Requirement R3 lets us avoid calling %
\LNDP algorithm $\cA$ 
for each of the $q$ query graphs,
which would result in a $\sqrt{n}$ %
blowup in privacy parameters. 
We circumvent this by adapting the ``mix-and-match'' technique \cite{EdenLRS25} to our reduction from inner product queries. 
The idea is that 
all nodes in the query graphs whose neighbors depend on the sensitive vector 
$\sv$---namely, those in $U_1$ and $U_2$---have at most \emph{two} views across {\em all} query graphs. %
The neighborhood of each node in $U_2$ is the same for all query graphs, whereas each node $v_{(1,i)}$ in $U_1$ has one of two possible neighborhoods, determined by the bit $a_i$ of the inner product query. 
To obtain the information that $\cB$ needs from each node in $U_1$ and $U_2$, we run the randomizers assigned to each node by $\cA$ on the two query graphs where $\iq$ is set to $\vec 0$ and $\vec 1$, and cache 
the two reports from each node.
All other nodes %
have no access to the information pertaining to the secret vector 
$\sv$, so we can run randomizers on them multiple times as part of postprocessing with no privacy loss. %
By appropriately matching the reports from different nodes so that they are %
what each node would report in the true query graph for a given query, %
we can simulate each of the $q$ runs of $\cA$ with only two {\em private} calls to %
$\cA$.

\subsection{Related Work}
\label{sec:unrelated}

DP for graph data was first investigated by \cite{NissimRS07}, who gave edge-private algorithms for triangle counting and estimating the cost of the minimum spanning tree. 
General private subgraph counts were first studied by 
\cite{KarwaRSY14} (strengthening \cite{RastogiHMS09}, who studied the problem under a relaxation of DP), in the central model under edge privacy.

\paragraph{Node DP.}
The first node-private algorithms were for subgraph counts and appeared in
\cite{BlockiBDS13,ChenZ13,KasiviswanathanNRS13}, after node privacy was formulated by \cite{HayLMJ09}. Subsequent work developed node-private algorithms for degree distributions via Lipschitz extensions \cite{RaskhodnikovaS16} and projections \cite{DayLL16}, as well as for practical triangle counting \cite{DingZBJ18}, graphon and $\ergraph$ parameter estimation \cite{BorgsCSZ18,SealfonU21}, counting connected components \cite{KalemajRST23, JainSW24}, and node privacy under continual observation \cite{JainSW24}. More recent work %
used sum-of-squares methods for node-private graphon estimation and edge-density estimation in random graph models \cite{ChenDDHLS24,ChenDHS24}. Generic, exponential-time private algorithms with accuracy based on a function's down-sensitivity have also been applied to node privacy \cite{FangDY22,LinderRSS25}.
Node-private methods have also been applied to epidemiological data \cite{HodNGKTS26}.
Closest to our setting, \cite{RaskhodnikovaSWZ26} introduced local node differential privacy and studied noninteractive and interactive algorithms for releasing functions of a graph's degree list. See \cite{RaskhodnikovaS16-E} for a survey of graph privacy.

\paragraph{(Vanilla) Local DP.}
Local randomization algorithms \cite{Warner65,EvfimievskiGS03} predate the
introduction of DP, but were later recognized as
mechanisms satisfying this definition~\cite{DworkMNS16,MishraS06}.
The (interactive) local model of DP was 
formalized and studied systematically by \cite{KasiviswanathanLNRS11}. 
Lower bounds for bit summation in this model were established by \cite{BeimelNO08}.
A large body of subsequent work
has
developed the statistical and information-theoretic foundations of local privacy, including general lower-bound techniques and tight characterizations for basic estimation tasks
\cw{(e.g.,}
\cite{DuchiJW13}). Interactivity in the local model was systematically studied by \cite{JosephMNR19}, who
formalized several notions of interactivity and
showed that, for hypothesis testing, noninteractive randomized response is optimal even among interactive protocols. Local DP also saw several large-scale deployments, including RAPPOR \cite{ErlingssonPK14}, Prochlo \cite{BittauEMMRLRKTS17}, Microsoft telemetry \cite{DingKY17}, and Apple's method for identifying popular image locations \cite{apple-iconic-scenes23}.
While some works \cite{KasiviswanathanLNRS11,BeimelNO08,JosephMNR19}
have dealt %
with the interactive setting, many---including ours---focus
on the noninteractive setting (e.g.,
\cite{DuchiJW13,BassilyS15,BunNS19,EdmondsNU20,ChenGKM21,FeldmanMRT25,CanonneGS26,AliakbarpourFRS26}).

\paragraph{Local Edge DP.}
With the exception of \cite{RaskhodnikovaSWZ26}, all prior work on local privacy for graphs considered the weaker notion of local {\em edge} differential privacy (LEDP). The model was introduced by \cite{QinYYKXR17}, who studied the generation of synthetic graphs, and was formally defined for the noninteractive and interactive settings in \cite{ImolaMC21} and \cite{EdenLRS25}, respectively. %
Subsequent work investigated the estimation of graph statistics, including triangle and $k$-star counts \cite{ImolaMC21,ImolaMC22-tri,EdenLRS25}, $k$-core decomposition and densest subgraph computation \cite{DhulipalaLRSSY22,MundraPSL25}, and more general subgraph counting problems \cite{SuppakitpaisarnPHH25}. Several works established lower bounds for these tasks. In particular, \cite{EdenLRS25} introduced the ``mix-and-match'' technique, obtaining lower bounds for triangle counting, and \cite{SuppakitpaisarnPHH25} extended this approach to complete subgraphs. Most recently, \cite{HillebrandSS25} established optimal error bounds for triangle counting in constant-degeneracy graphs. %

\subsection{Organization}

\Cref{sec:prelims} formalizes the subgraph-counting problem and gives background on differential privacy and the \LNDP model.
\Cref{sec:ub-sub} presents our \LNDP algorithm for subgraph counting (\Cref{alg:subh}), along with its error guarantees in \Cref{thm:sub-count}.
\Cref{sec:lb} gives our lower bound on the error required for subgraph counting under \LNDP (\Cref{thm:lb}), establishing the optimality of our algorithm's error.

\section{Preliminaries}
\label{sec:prelims}

We let $[n]$ denote $\set{1,\ldots, n}$.
Additionally, for an undirected graph $G$ on node set $[n]$, the set $\neigh{G}{i}$ denotes the neighborhood of node $i \in [n]$ in $G$.

We are interested in counting the number of \emph{unlabeled}, \emph{noninduced} instances of some
template graph $H$ in an undirected input graph $G$, under \LNDP.
(The number of \emph{labeled} instances in $G$ is a factor of $|\Aut(H)|$ larger than the number of unlabeled instances, where we let $|\Aut(H)|$ denote the number of automorphisms of $H$.)

\begin{definition}[Subgraph counting]
\label{def:subgraph-count}
Let $\subgraph = (V(\subgraph), E(\subgraph))$ be an undirected graph on $k$ nodes, and $n \in \mathbb{N}$.
An \emph{unlabeled, noninduced
 subgraph count query} on graphs with vertex set $[n]$ is specified by $\subgraph$ and maps an undirected input graph $G$ on vertex set $[n]$ to the number of unlabeled copies of $H$ in $G$. %
In other words, using $\hookrightarrow$ to denote an injection, it returns
\[
\#H(G) = \frac{1}{|\Aut(H)|}
\Big|\big\{\phi:V(H)\hookrightarrow V(G):
\forall
\set{u,v}
\in E(H),\
\set{\phi(u),\phi(v)}
\in E(G)\big\}\Big|\, .
\]
\end{definition}

\subsection{Local Node Differential Privacy}
\label{sec:lndp}

We now recall the definitions of DP and \LNDP (see \Cref{sec:dp-background} for additional background).

\paragraph{Node neighbors and differential privacy.}%

Differential privacy is defined with respect to neighboring datasets, which differ in the data of a single individual. In the tabular setting, two datasets $x,x'\in\cX^n$ are neighbors if they differ in one entry. We focus on node privacy for graphs: graphs $G$ and $G'$ on node set $[n]$ are \emph{node neighbors}, denoted $G\sim G'$, if one can be obtained from the other by rewiring a single node, i.e., by changing only the edges incident to some node $i\in[n]$. More generally, $G$ and $G'$ are at \emph{node distance} $k$ if one can be obtained from the other by rewiring $k$ nodes.

\begin{definition}[$(\eps,\delta)$-indistinguishability]
\label{def:indistinguishable}
Let $\eps>0$ and $\delta\in[0,1]$. Two distributions $P,Q$ over outcome space $\cY$ are \emph{$(\eps,\delta)$-indistinguishable}, denoted $P\simed Q$, if for all $Y\subseteq\cY$, %
$$
    P(Y)\le e^\eps Q(Y)+\delta
    \qquad\text{and}\qquad
    Q(Y)\le e^\eps P(Y)+\delta.
$$
We %
write $R_1\simed R_2$ for random variables $R_1$ and $R_2$ with $(\eps,\delta)$-indistinguishable distributions.
\end{definition}

\begin{definition}[Differential privacy (DP) \cite{DworkMNS16}]
\label{def:dp batch}
Let $\eps>0$ and $\delta\in[0,1]$. A randomized algorithm $\cA:\cU^*\to\cY$ is \emph{$(\eps,\delta)$-DP} if $\cA(x)\simed \cA(x')$ for every pair of neighboring inputs $x,x'\in\cU^*$.
\end{definition}

Many DP algorithms in the local model are built from the \emph{randomized response} primitive of \cite{Warner65}, which was shown by \cite{DworkMNS16} to satisfy $(\eps,0)$-DP.

\begin{definition}[Randomized response]
\label{def:rr-alg}
    Given $\eps > 0$, let $\rr[\eps]:\zo \to \zo$, called \emph{randomized response}, be the $(\eps,0)$-DP algorithm that, on input $b\in\zo$, returns $b$ with probability $\frac{e^\eps}{e^\eps + 1}$, and $1-b$ with probability $\frac{1}{e^\eps + 1}$.
\end{definition}

The standard method for debiasing randomized response is analyzed, e.g., in \cite{EdenLRS25}.

\begin{fact}
\label{fact:RR}
    Given output $X = \rr[\eps](b)$, the \emph{debiased} estimate
    \(
        Y = \frac{X(e^\eps + 1) - 1}{e^\eps - 1}
    \)
    has expectation $\E[Y] = b$. Moreover, $\Var(Y) = \frac{e^\eps}{(e^\eps - 1)^2}$,
    which is
    $\Theta \paren{\frac{1}{\eps^2}}$ for $\eps \leq 1$.
\end{fact}

\paragraph{LNDP$^\star$.}
Following the convention in \cite{RaskhodnikovaSWZ26}, we use \LNDP to denote noninteractive LNDP. In the \LNDP setting, each node runs a \emph{local randomizer} on its edge list
and releases the result. (Local randomizers can depend on public randomness.) These noisy releases are then postprocessed to obtain a final result. \LNDP requires that the joint distributions over all nodes' reports be $(\eps,\del)$-indistinguishable on node-neighboring input graphs, for all strings $\pubrand$ of public randomness.

\begin{definition}[%
\LNDP %
\cite{RaskhodnikovaSWZ26}]
\label{def:lndp}
    Let $\eps > 0, \del\in[0,1]$, and $n\in \N$.
    An %
    algorithm $\cA$
    is \edLNDP
    if there exist (i) a distribution $\distpubrand$ over  
    strings $\pubrand$ (public randomness), 
    (ii) local randomizers $\cR_{1,\pubrand},\ldots, \cR_{n,\pubrand}$ (depending on $\pubrand$),
    and (iii) a postprocessing algorithm $\cP$ such that 
    \begin{enumerate}
        \item \emph{(Locality and noninteractivity)}
        for all graphs $G$ on node set $[n]$, the algorithm
        $\cA$ can be represented as
        $
            \cP \bigl( \cR_{1,\pubrand}(\neigh{G}{1}), \ldots, \cR_{n,\pubrand}(\neigh{G}{n}) \bigr),
        $
        where $\pubrand\sim \distpubrand$ and node $i$ runs a randomizer $\cR_{i,\pubrand}$ on its neighborhood $\neigh{G}{i}$, and 
        \item \emph{(Node privacy)} For all node-neighboring graphs $G$ and $G'$ on node set $[n]$ and all settings of public randomness $\pubrand$,
        the distributions of vectors of outputs released by the randomizers are $(\eps,\del)$-indistinguishable---that is, $\vec{\cR}_\pubrand(G)
        \approx_{\eps,\del}
        \vec{\cR}_\pubrand(G')$, where $\vec{\cR}_\pubrand(G) = \bigl( \cR_{1,\pubrand}(\neigh{G}{1}), \ldots, \cR_{n,\pubrand}(\neigh{G}{n}) \bigr)$.
        Equivalently, for all $\rho$, the map $\vec{\cR}_\pubrand$ is $(\eps, \del)$-DP under the node-neighboring relation.
    \end{enumerate}
\end{definition}

\section{Algorithm for Subgraph Counting with \texorpdfstring{\LNDP}{LNDP}}
\label{sec:ub-sub}

In this section we present our \LNDP algorithm for counting the number of unlabeled instances of a $k$-node template graph $H$ 
in an $n$-node graph $G$. Our algorithm (separately) handles counting induced and noninduced copies of $H$.

For a template graph $H = ([k], E(H))$, let $\vc{H} \in [k]$ denote the size of the minimum vertex cover of $H$.

\begin{theorem}[Subgraph counting]
\label{thm:sub-count}
    Let $\algsubh$ be \Cref{alg:subh} with parameters $\eps\in (0,1]$, 
    $k,n\in \N$,
    and undirected template graph $H$ on node set $[k]$.
    Then $\algsubh$ is $(\eps,0)$-\LNDP and, for all graphs $G$ on node set $[n]$, satisfies
    \[
        \Ex\bracks{\babs{\algsubh(G) - \subgraphcount(G)}}
        =
        O\paren{\frac{\sqrt{k!}}{|\Aut(H)|} \cdot n^{k-\frac12}\cdot \paren{\frac{\vc{H}}{\eps}}^{{\vc{H}}} }
        = O_k\paren{ \frac{n^{k-\frac12}}{\eps^{k-1}} }.
    \]
\end{theorem}

\Cref{alg:subh} can be modified to instead count \emph{induced} copies of $H$.
The resulting error guarantee is the same, except $\vc{H}$ is replaced by $k-1$.

\paragraph{Algorithm overview.}
At a high level, we first randomly partition the node set into blocks of $k$ nodes.
For each block, we randomly assign each node to a node in $H$, and have each node in a minimum vertex cover $\msf{C}$ of $H$ use randomized response to release a debiased version of a noisy bit indicating whether its edge set matches that of its assigned node in $H$.
Multiplying these debiased values yields a noisy indicator for whether $H$ appears in the subgraph induced by the block (for a random labeling of nodes in the block).
Finally, we sum these noisy indicators 
and scale up by $\frac{k!}{|\Aut(H)|}\cdot\binom{n-1}{k-1}$ to obtain an unbiased estimate for the total number of unlabeled occurrences of $H$ in $G$.

In our algorithm, we use the following notation. For a sequence 
$I=(v_1,\ldots, v_k)$ of nodes in $[n]$,
define the  indicator variable $\edgematch^{I,H}_i = 1$ iff for every neighbor $j$ of $i$ in $H$, the pair $\set{v_i,v_j}$ forms an edge in $G$.

\input{algs/lndp-counting}

\begin{proof}[Proof of \Cref{thm:sub-count}]
    Assume, without loss of generality, that $n$ is divisible by $k$ so that the step on Line~\ref{line:pub-rand-alg} is well defined; otherwise, add dummy nodes to $G$ to ensure $k$ divides $n$.

    \emph{(Privacy.)}
    Fix a random permutation $P$ and a pair of node-neighboring graphs $G$ and $G'$, where $i^*$ is the rewired node.
    By \Cref{def:lndp}, to prove that the algorithm is $(\eps, 0)$-\LNDP, 
    it suffices to show that releasing the vector of node outputs
    $\bparen{
    \cw{Y}_i^{(j)}}_{i\in \msf{C} ,j\in[n/k]}$
    is $(\eps, 0)$-DP.
    This vector is a function only of the sequence 
    $\bparen{\edgematch^{P^{(j)},H}_i}_{i\in\msf{C}, j\in[n/k]}$.
    The only output-producing nodes that may have different values of $\edgematch^{P^{(j)},H}_i$ for $G$ and $G'$ are the $\vc{H}$
    nodes in the sequence $P^{(j)}$ containing $i^*$.
    Thus, $\bparen{\edgematch^{P^{(j)},H}_i}_{i\in \msf{C},j\in[n/k]}$ has $\ell_1$ sensitivity at most $\vc{H}$,
    so by group privacy the algorithm that releases
    each bit of this vector via debiased randomized response with parameter $\eps'$ (as on Line~\ref{line:sub-node-release}) is $(\eps,0)$-DP.
    
    \emph{(Accuracy.)} 
    Let $T$ denote the number of unlabeled copies of $H$ in $G$, and define the indicator $S^{(j)} = \prod_{i\in \msf{C}} \edgematch_i^{P^{(j)},H}$; that is,
    $S^{(j)}=1$ iff the specified labeling embeds $H$ as a subgraph of the graph induced by $P^{(j)}$.
    We first show that $\wh T$ is an unbiased estimator of~$T$. For each $j \in [n/k]$, 
    \begin{align}\label{eq:expectation-of-Sj}
        \Ex[S^{(j)}] =\Pr[S^{(j)} = 1] =  \frac{T}{\binom{n}{k}} \cdot \frac{|\Aut(H)|}{k!},
    \end{align}
since  $S^{(j)} = 1$ if one of the $T$ copies of $H$ lands in $P^{(j)}$ and the right labeling is chosen.
From the debiasing on Line~\ref{line:sub-debias} and mutual independence of the $Y_i^{(j)}$ for fixed permutation $P$, we have 
\(
    \Ex[\wh S^{(j)} \mid P] = \prod_{i\in\msf{C}} \Ex[Y_i^{(j)} \mid P]
    = \E[S^{(j)} \mid P],
\)
so $\Ex[\wh S^{(j)}] = \Ex[S^{(j)}]$ by the law of total expectation.
By the definition of $\wh T$ on Line~\ref{line:subh-ans} and linearity of expectation,
\[
    \Ex[\wh T]
    = \frac{k!}{|\Aut(H)|} \cdot \binom{n-1}{k-1} \sum_{j \in [n/k]} \Ex[\wh S^{(j)}]
    = \frac{k!}{|\Aut(H)|}\binom{n-1}{k-1} \cdot \frac{n}{k} \cdot  \frac{T}{\binom{n}{k}} \cdot \frac{|\Aut(H)|}{k!}
    = T.
\]
   
    We next bound the algorithm's additive error.
    The error arises from two sources: randomized response (privacy error) and random permutation sampling (sampling error).
    Let $\wt T = \frac{k!}{|\Aut(H)|}\cdot \binom{n-1}{k-1}\cdot \sum_{j\in[n/k]} S^{(j)}$.
    By the triangle inequality, the error can be bounded as:
    \asnote{For any two real-valued random variables $A,B$, we have $\Var(A+B) = \Var(A) + \Ex_{a\sim A}\paren{\Var(B|A=a)}$. So applying that here with $A=$(sampling err) and $B=$(privacy err) might remove some of the redundancy in the calculations below? It would save a few lines at most...}
    \begin{equation}
    \label{eq:subh-error-tri-ineq}
        \Ex \bbracks{\babs{\wh T - T}}
        \leq
        \underbrace{\Ex \bbracks{\babs{\wh T - \wt T}}}_{\text{privacy error}}
        +
        \underbrace{\Ex \bbracks{\babs{\wt T - T}}}_{\text{sampling error}}.
    \end{equation}
    
\emph{Bounding the privacy error $\Ex\bbracks{\babs{\wh T - \wt T}}$.} We first bound the expected privacy error, conditioned on the public permutation $P$; since this bound is independent of $P$, taking expectation over $P$ yields the same bound unconditionally.
For each $j\in[n/k]$, let
$\Delta_j := \wh S^{(j)}-S^{(j)}$. To analyze the privacy error, we take advantage of the fact that, conditioned on $P$, all $\Delta_j$'s are independent (because local randomizers perform randomized response independently) and also, by the properties of the debiasing step (see \Cref{fact:RR}), have expectation~0. 

By linearity of conditional expectation,
\begin{align}
\label{eq:priverr-scale}
    \E\bracks{\abs{\wh T-\wt T}\mid P}
    &=
    \frac{k!}{|\Aut(H)|}\binom{n-1}{k-1}
    \E\Big[\Big|\sum_{j\in[n/k]} \Delta_j \Big| \ \big|\ P\Big].
\end{align}
To bound the expectation in \Cref{eq:priverr-scale}, we apply Jensen's inequality and then use the
fact that, conditioned on $P$, the $\Delta_j$'s are independent and have mean zero:
\begin{align}
\label{eq:priverr-jensen}
    \E\Big[\Big|\sum_{j\in[n/k]} \Delta_j \Big| \ \big|\ P\Big]
    &\le
    \sqrt{\E\Big[\Big(\sum_{j\in[n/k]} \Delta_j\Big)^2 \ \big|\ P\Big]}
    =
    \sqrt{\sum_{j\in[n/k]}\E[\Delta_j^2\mid P]} .
\end{align}
For each block $j\in[n/k]$, since $\E[\Delta_j\mid P]=0$ and $S^{(j)}$ is fixed conditioned on $P$, we get
\begin{align}
\label{eq:priverr-block}
    \E[\Delta_j^2\mid P]
    &=  \Var(\Delta_j\mid P)
    = \Var(\wh S^{(j)}\mid P)
    \le
    \E[(\wh S^{(j)})^2\mid P] \nonumber\\
    &=
    \prod_{i\in\msf C}\E[(Y_i^{(j)})^2\mid P]
    =
    O\paren{\frac{1}{(\eps')^{2\vc{H}}}}
    =
    O\paren{\paren{\frac{\vc{H}}{\eps}}^{2\vc{H}}},
\end{align}
using respectively, in the final line of \eqref{eq:priverr-block}, independence of the randomized-response noise added by different nodes within the block, \Cref{fact:RR}, and $\eps'=\eps/\vc{H}$.

Combining \Cref{eq:priverr-scale,eq:priverr-jensen,eq:priverr-block} and then using $\binom{n-1}{k-1}\leq\frac{(n-1)^{k-1}}{(k-1)!}$ gives
\begin{align}
    \E\bracks{\abs{\wh T-\wt T}\mid P}
    &\le
    \frac{k!}{|\Aut(H)|}\binom{n-1}{k-1}\,
     \sqrt{\frac nk}
    \cdot
    O\paren{\paren{\frac{\vc{H}}{\eps}}^{\vc{H}}}\nonumber\\
    &= O\paren{\frac{\sqrt{k}}{|\Aut(H)|} \cdot n^{k-\frac12}\cdot \paren{\frac{\vc{H}}{\eps}}^{\vc{H}}}
    . \label{eq:privacy-error}
\end{align}
By the law of total expectation, taking expectation over $P$ gives the same bound unconditionally.

\emph{Bounding the sampling error $\Ex\bbracks{\babs{\wt T - T}}$.}
We start by applying Jensen's inequality and then using the fact that $\E[\wt T]=T$:
\begin{equation}
\label{eq:jensen-sample}
    \Ex\bbracks{\babs{\wt T- T}}
    = \Ex\bbracks{\sqrt{(\wt T-T)^2}}
    \le \sqrt{\Ex\bbracks{(\wt T-T)^2}} 
    =\sqrt{\Var(\wt T)}.
\end{equation}
To bound $\Var(\wt T)$, we represent $\wt T$ as $\wt T=\frac S{p}$, where $S:=\sum_{j\in[n/k]}S^{(j)}$ and $p := \frac{|\Aut(H)|}{k!\binom{n-1}{k-1}}$.

The %
summands $S^{(j)}$ in the expression for $S$ are not independent. Nevertheless, since their pairwise correlation is fairly weak, one can still bound $\Var(S) = \Ex[S^2]-\Ex[S]^2$.
By~\eqref{eq:expectation-of-Sj}, we get $\Ex[S]=\E[\sum_{j\in[n/k]}S^{(j)}]=\frac n k \cdot \frac{T}{\binom{n}{k}} \cdot \frac{|\Aut(H)|}{k!}=Tp.$ Now we expand the second moment:
\begin{align*}
    \Ex[S^2]
    & =  \Ex\big[(\sum_{j\in[n/k]}S^{(j)})^2\big]
    =
    \sum_{j\in[n/k]} \Ex[S^{(j)}]
    +
    \sum_{i\neq j}\Ex[S^{(i)}S^{(j)}] 
    = T p + \sum_{i\neq j} \Pr[S^{(i)}=S^{(j)}=1].
\end{align*}
The last sum has $\frac n k(\frac n k-1)$ terms. Blocks $P^{(i)}$ and $P^{(j)}$ in the partition are vertex-disjoint, so both $S^{(i)}$ and $S^{(j)}$ can be 1 only if the two blocks sample two vertex-disjoint copies of $H$ in $G$. Let $D$ be the number of ordered pairs of vertex-disjoint copies of $H$. Then $D\le T^2$, and
\begin{align*}
\sum_{i\neq j} \Pr[S^{(i)}=S^{(j)}=1]
&=\frac {n(n-k)} {k^2} \cdot D\frac{|\Aut(H)|^2}
    {(k!)^2\binom{n}{k}\binom{n-k}{k}}
\leq T^2 \frac{|\Aut(H)|^2}
    {(k!)^2\binom{n-1}{k-1}\binom{n-k-1}{k-1}}\\
& =
    p^2 \cdot
    \frac{\binom{n-1}{k-1}}{\binom{n-k-1}{k-1}}
    = p^2 \cdot 
    \prod_{r=1}^{k-1}
    \left(1+\frac{k}{n-k-r}\right)
    =
    p^2\left(1+O\paren{\frac{k^2}{n}}\right).
    \end{align*}
    Therefore,
    \(
        \Var(S)=\Ex[S^2]-\Ex[S]^2
        \leq T p
        +
        T^2p^2\paren{1+O\paren{\frac{k^2}{n}}}
        -
        T^2p^2 
        =
        T p + O\paren{\frac{k^2 T^2 p^2}{n}}.
    \) 

To bound $\Var(\wt T)$, recall that $\wt T =\frac S{p}$. 
Using $T=O(\frac{n^k}{|\Aut(H)|})$ and $\frac 1{p}=O(\frac{k!n^{k-1}}{|\Aut(H)|})$, we get
\begin{align*}
    \Var(\wt T)
    =
    \frac{\Var(S)}{p^2}
    &\le
    \frac{T}{p}
    +
    O\paren{\frac{k^2T^2}{n}}
    =
    O\paren{\frac{k!\cdot n^{2k-1}}{|\Aut(H)|^2}}.
\end{align*}
Substituting the expression for $\Var(\wt T)$ into \eqref{eq:jensen-sample} gives us
\[
    \Ex\bbracks{\babs{\wt T- T}}
    \le
    \sqrt{\Var(\wt T)}
    =
    O\paren{\frac{\sqrt{k!}}{|\Aut(H)|}\cdot n^{k-\frac12}}.
\]

    \emph{Bounding the overall error $\Ex\bbracks{\babs{\wh T - T}}$.}
    Substituting the bounds on privacy and sampling error into \Cref{eq:subh-error-tri-ineq} yields nearly the same bound as in \eqref{eq:privacy-error}, except that $\sqrt k$ changes to~$\sqrt{k!}$:
    \[
        \Ex\bbracks{\babs{\wh T - T}}
        = O\paren{\frac{\sqrt{k!}}{|\Aut(H)|} \cdot n^{k-\frac12}\cdot \paren{\frac{\vc{H}}{\eps}}^{\vc{H}} }.
    \]
This completes the proof of \Cref{thm:sub-count}.
\end{proof}

\input{sections/lower-bound-polynomial}

\AtNextBibliography{\small}
\printbibliography[heading=bibintoc]

\appendix
\section*{Appendix}
\addcontentsline{toc}{section}{Appendix}

\section{Background on Differential Privacy}
\label{sec:dp-background}

This section collects useful background on differential privacy, including common mechanisms and standard properties of the definition.

A common way to release a function's value under differential privacy is to use noise scaled to its sensitivity.

\begin{definition}[Sensitivity] Let $k\in\mathbb{N}$ and $f:\cX \rightarrow \mathbb{R}^k$
be a function. Let $p\in\{1,2\}$. 
The \emph{$\ell_p$-sensitivity of $f$}, denoted by $\Delta_p$, is defined as
    $\Delta_p=\max_{x\sim y}\|f(x)-f(y)\|_p.$
\end{definition}

\begin{lemma}[Laplace mechanism~\cite{DworkMNS16}] \label{lem:Laplacemech} Let $k\in\N$ and $\eps>0$ and $f:\cX^n \rightarrow \mathbb{R}^k$ 
be a function
with $\ell_1$-sensitivity $\Delta_1$. 
The Laplace mechanism is defined as
$\cA(x)=f(x)+(Z_1,\dots,Z_k)$, where the $Z_i \sim 
\Lap(\Delta_1/\eps)$ are independent for all $i\in[k]$. The Laplace mechanism is $(\eps,0)$-DP.
\end{lemma}

Differential privacy is robust to postprocessing (i.e., the application of an arbitrary randomized algorithm to the output of a differentially private algorithm), and its parameters degrade gracefully under composition.

\begin{lemma}[DP is robust to postprocessing \cite{DworkMNS16}]
\label{lem:post-proc}
Let $\cM:\cU^*\to \cY$ be a randomized algorithm that is $(\eps,\del)$-DP. Let $f:\cY\to \cZ$ be a randomized function. Then $f\circ \cM:\cU^*\to \cZ$ is $(\eps,\del)$-DP.
\end{lemma}

\begin{lemma}[Basic composition~\cite{DworkL09,DworkRV10,DworkKMMN06}]
\label{lem:basic-comp}
    Let $\eps_1,\eps_2> 0$ and $\delta_1,\delta_2\in[0,1)$.  
    Let $\cA_1 \colon \cU^* \to \cY$ be an $(\eps_1, \delta_1)$-DP algorithm and let $\cA_2 \colon \cU^* \times \cY \to \cZ$ be an $(\eps_2, \delta_2)$-DP algorithm.
    Then for all $x\in\cU^*$, algorithm $\cA_2(x,\cA_1(x))$ is $(\eps_1+\eps_2,\delta_1+\delta_2)$-DP. 
\end{lemma}

\section{Error Required for Private Inner Product Queries}
\label{sec:lem-inn-prod-explanation}

We provide a proof sketch of \Cref{lem:inn-prod-dp-lb-new}, restated below for convenience.

\InnerProductDPLemma*

\begin{proof}[Proof sketch of \Cref{lem:inn-prod-dp-lb-new}]
\cite[Theorem 23]{DworkMT07} considers inner product queries $a \in \{\pm 1\}^n$  and shows that, for sufficiently large $n$, there exist $m = \Theta(n)$ and  $\alpha = \Omega(\sqrt{n})$ such that the following holds with probability at least $1 - e^{-cn}$ for some constant $c > 0$.
If an algorithm $\cB$ answers at least a $\frac{4}{5}$ fraction of uniformly random inner product queries with additive error at most $\alpha$, then there is an efficient attacker $\cA$ that reconstructs at  least a $\frac{99}{100}$ fraction of the bits of a uniformly random input $x \in \zo^n$ from the outputs of $\cB$ alone. Since no adversary can reconstruct a $\frac{99}{100}$ fraction of the bits in a uniformly random input based only on the outputs of a $(4, \frac{1}{25})$-DP algorithm, \Cref{lem:inn-prod-dp-lb-new} follows. Translating from the \cite{DworkMT07} setting of $a\in \set{\pm 1}^n$ to our setting of $a\in \zo^n$ only introduces a factor of $2$ in the $\alpha$ term.
\end{proof}

\end{document}

%% file: figures/table-anatoly-intro.tex
\begin{table}[tb]
\centering
\renewcommand{\arraystretch}{1.2}
\setlength{\tabcolsep}{6pt}

\begin{tabular}{
    c
    ||
    c|
    c
}
\toprule
\rowcolor{gray!30}
{\small \twoline{\textbf{Privacy}}{\textbf{guarantee}}}
& {\small \textbf{Edge DP}}
& {\small \textbf{Node DP}}
\\
\midrule
\rowcolor{gray!10}
{\small \textbf{Central model}}
    & \twoline{$\Theta_{\eps,\del}\left(n^{k-2} \right)$}{\footnotesize\cite{NissimRS07}}\Tstrut\Bstrut
    & \twoline{$\Theta_{\eps,\del} \left(n^{k-1} \right)$}{\footnotesize\cite{KasiviswanathanNRS13,ChenZ13}}\Tstrut\Bstrut \\
{\small \textbf{\twoline{Noninteractive}{local model}}}
    & \twoline{$\Theta_{\eps,\del} \left(n^{k-1} \right)^\dagger$}{\footnotesize\cite{EdenLRS25, SuppakitpaisarnPHH25}}\Tstrut\Bstrut
    & \cellcolor{blue!8}\twoline{$\Theta_{\eps,\del} \left(n^{k-\frac{1}{2}} \right)$}{\footnotesize \textbf{Our work}}\Tstrut\Bstrut \\
\bottomrule
\end{tabular}

\caption{Additive error for counting noninduced $k$-node subgraphs in $n$-node graphs %
for fixed $k$.
All algorithms are $(\eps, 0)$-DP in the respective models, and lower bounds hold for $\eps \leq 1$ and $\del \leq \frac{1}{1000}$. The notation $\Theta_{\eps, \del}( \cdot)$ hides dependence on privacy parameters $\eps, \del$.\\
$^\dagger$LEDP lower bounds are only known for cliques and cycles, not general subgraphs. 
Interactive LEDP algorithms for subgraph counting have the same worst-case dependence on $n$ (but, in some cases, better dependence on $\eps$) \cite{ImolaMC21,HillebrandSS25}. %
}

\label{tab:sub-results}

\end{table}

%% file: figures/query-graph-new.tex
\begin{figure}
\centering
\resizebox{0.65\linewidth}{!}{%
\begin{tikzpicture}[
    vertex/.style={circle, draw=black, fill=red, inner sep=0pt, minimum size=0.25cm},
    redvertex/.style={circle, draw=black, fill=blue!70, inner sep=0pt, minimum size=0.25cm},
    setlabel/.style={font=\large\bfseries, color=red},
    wlabel/.style={font=\large\bfseries, color=blue},
    infobox/.style={
        rectangle,
        draw=none,
        fill=blue!30,
        text=black,
        align=center,
        font=\sffamily\footnotesize,
        inner sep=3pt,
        opacity=0.85
    },
    scale=0.875
]

    \node[setlabel, anchor=east] at (-3.55,1.7) {$U_2$};
    \foreach \i [count=\x] in {1, 2, 3, 4, 5, 6} {
        \node[vertex] (v\x) at (\x*0.65 - 3.5, 1.7) {};
    }
    \node[above=0.1cm] at (v1) {\small $\node{2}{1}$};
    \node[above=0.1cm] at (v4) {\small $\node{2}{i'}$};
    \node[above=0.1cm] at (v6) {\small $\node{2}{n}$};

    \node[setlabel, anchor=east] at (-3.55,-2) {$U_1$};
    \foreach \i [count=\x] in {1, 2, 3, 4, 5, 6} {
        \node[vertex] (u\x) at (\x*0.65 - 3.5, -2) {};
    }
    \node[below=0.1cm] at (u1) {\small $\node{1}{1}$};
    \node[below=0.1cm] at (u3) {\small $\node{1}{i}$};
    \node[below=0.1cm] at (u6) {\small $\node{1}{n}$};

    \node[wlabel] at (3.5,-0.2) {$U_3$};
    \foreach \i [count=\x] in {1, 2, 3, 4, 5} {
        \node[redvertex] (w\x) at (1.75, 1.3 - \x*0.5) {};
    }

    \node[right=0.1cm] at (w1) {\small $\node{3}{1}$};
    \node[right=0.1cm] at (w2) {\small $\node{3}{2}$};
    \node[right=0.2cm] at ($(w2)!0.4!(w5)$) {\small $\vdots$};
    \node[right=0.1cm] at (w5) {\small $\node{3}{n}$};

    \node[infobox] (boxA) at (3.85,1.7)
        {$U_2$ and $U_3$ form\\ complete bipartite graph};
    \draw[->, thin] (boxA.west) -- (1.25,1.10);

    \foreach \x in {1, ..., 5} {
        \draw[darkgray] (v4) -- (w\x);
    }

    \foreach \x in {1, ..., 5} {
        \draw[black, thick, dashed] (u3) -- (w\x);
    }

    \node[infobox] (boxB) at (3.85,-2)
        {$\{\node{1}{i}, \node{3}{j}\} \in E$\\ for all $j$ if $a_i = 1$};
    \draw[->, thin] (boxB.west) -- (1.25,-1.7);

    \foreach \i [count=\x] in {1, 2, 3, 4, 5, 6} {
        \draw[red, thick] (u3) -- (v\x);
    }

    \node[infobox] (boxX) at (-5,0)
        {$\{\node{1}{i}, \node{2}{j}\} \in E$\\ for all $j$ if $x_i=1$};
    \draw[->, thin] (boxX.east) -- (-2.4,0);

\end{tikzpicture}%
}
\caption{The query graph (\Cref{def:query-graph-new})
for $H = K_3$. Red edges depend on the sensitive vector $\sv$; dashed edges depend on the query vector $\iq$; the remaining edges are fixed.}
\label{fig:querygraph}
\end{figure}
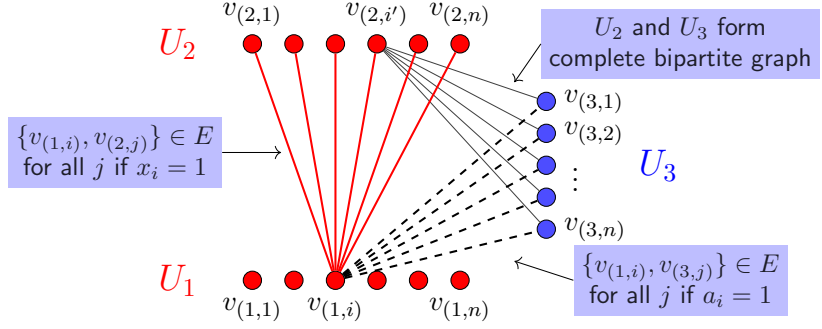

%% file: figures/c4-query-graph.tex
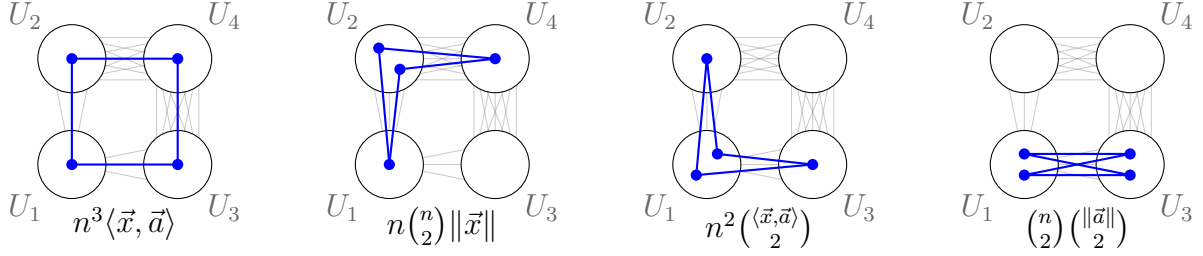
\begin{figure}
\centering
\begin{tikzpicture}[scale=0.7]

\begin{scope}[shift={(0,0)}]

  \coordinate (p12_1a) at (-0.4, 0.00);
  \coordinate (p12_1b) at ( 0.00, 0.00);
  \coordinate (p12_1c) at ( 0.4, 0.00);
  \coordinate (p12_2a) at (-0.4, 2.00);
  \coordinate (p12_2b) at ( 0.00, 2.00);
  \coordinate (p12_2c) at ( 0.4, 2.00);

  \coordinate (p13_1a) at ( 0.00,  0.4);
  \coordinate (p13_1b) at ( 0.00,  0.00);
  \coordinate (p13_1c) at ( 0.00, -0.4);
  \coordinate (p13_3a) at ( 2.00,  0.4);
  \coordinate (p13_3b) at ( 2.00,  0.00);
  \coordinate (p13_3c) at ( 2.00, -0.4);

  \coordinate (p24_2a) at ( 0.00,  2.4);
  \coordinate (p24_2b) at ( 0.00,  2.00);
  \coordinate (p24_2c) at ( 0.00,  1.6);
  \coordinate (p24_4a) at ( 2.00,  2.4);
  \coordinate (p24_4b) at ( 2.00,  2.00);
  \coordinate (p24_4c) at ( 2.00,  1.6);

  \coordinate (p34_3a) at ( 1.6,  0.00);
  \coordinate (p34_3b) at ( 2.00,  0.00);
  \coordinate (p34_3c) at ( 2.4,  0.00);
  \coordinate (p34_4a) at ( 1.6,  2.00);
  \coordinate (p34_4b) at ( 2.00,  2.00);
  \coordinate (p34_4c) at ( 2.4,  2.00);

  \foreach \s in {p12_1b} {
    \foreach \t in {p12_2a, p12_2b, p12_2c} {
      \draw[thin, gray, opacity=0.4] (\s)--(\t);
    }
  }
  \foreach \s in {p13_1b} {
    \foreach \t in {p13_3a, p13_3b, p13_3c} {
      \draw[thin, gray, opacity=0.4] (\s)--(\t);
    }
  }
  \foreach \s in {p24_2a, p24_2b, p24_2c} {
    \foreach \t in {p24_4a, p24_4b, p24_4c} {
      \draw[thin, gray, opacity=0.4] (\s)--(\t);
    }
  }
  \foreach \s in {p34_3a, p34_3b, p34_3c} {
    \foreach \t in {p34_4a, p34_4b, p34_4c} {
      \draw[thin, gray, opacity=0.4] (\s)--(\t);
    }
  }

  \node[draw, circle, fill=white, minimum size=0.9cm, label={[opacity=0.6, label distance=-3pt]above left:{$U_2$}}] at (0,2) {};
  \node[draw, circle, fill=white, minimum size=0.9cm, label={[opacity=0.6, label distance=-3pt]above right:{$U_4$}}] at (2,2) {};
  \node[draw, circle, fill=white, minimum size=0.9cm, label={[opacity=0.6, label distance=-3pt]below right:{$U_3$}}] at (2,0) {};
  \node[draw, circle, fill=white, minimum size=0.9cm, label={[opacity=0.6, label distance=-3pt]below left:{$U_1$}}] at (0,0) {};
  \coordinate (e1) at (0, 2);    %
  \coordinate (e2) at (2, 2);    %
  \coordinate (e3) at (2, 0);    %
  \coordinate (e4) at (0, 0);    %
  \draw[thick, color=blue] (e1) -- (e2);
  \draw[thick, color=blue] (e2) -- (e3);
  \draw[thick, color=blue] (e3) -- (e4);
  \draw[thick, color=blue] (e4) -- (e1);
  \fill[color=blue] (e1) circle (3pt);
  \fill[color=blue] (e2) circle (3pt);
  \fill[color=blue] (e3) circle (3pt);
  \fill[color=blue] (e4) circle (3pt);
  \node[below] at (1,-0.6) {$n^3 \langle \sv, \iq \rangle$};
\end{scope}

\begin{scope}[shift={(6,0)}]
  \coordinate (p12_1a) at (-0.4, 0.00);
  \coordinate (p12_1b) at ( 0.00, 0.00);
  \coordinate (p12_1c) at ( 0.4, 0.00);
  \coordinate (p12_2a) at (-0.4, 2.00);
  \coordinate (p12_2b) at ( 0.00, 2.00);
  \coordinate (p12_2c) at ( 0.4, 2.00);

  \coordinate (p13_1a) at ( 0.00,  0.4);
  \coordinate (p13_1b) at ( 0.00,  0.00);
  \coordinate (p13_1c) at ( 0.00, -0.4);
  \coordinate (p13_3a) at ( 2.00,  0.4);
  \coordinate (p13_3b) at ( 2.00,  0.00);
  \coordinate (p13_3c) at ( 2.00, -0.4);

  \coordinate (p24_2a) at ( 0.00,  2.4);
  \coordinate (p24_2b) at ( 0.00,  2.00);
  \coordinate (p24_2c) at ( 0.00,  1.6);
  \coordinate (p24_4a) at ( 2.00,  2.4);
  \coordinate (p24_4b) at ( 2.00,  2.00);
  \coordinate (p24_4c) at ( 2.00,  1.6);

  \coordinate (p34_3a) at ( 1.6,  0.00);
  \coordinate (p34_3b) at ( 2.00,  0.00);
  \coordinate (p34_3c) at ( 2.4,  0.00);
  \coordinate (p34_4a) at ( 1.6,  2.00);
  \coordinate (p34_4b) at ( 2.00,  2.00);
  \coordinate (p34_4c) at ( 2.4,  2.00);

  \foreach \s in {p12_1b} {
    \foreach \t in {p12_2a, p12_2b, p12_2c} {
      \draw[thin, gray, opacity=0.4] (\s)--(\t);
    }
  }
  \foreach \s in {p13_1b} {
    \foreach \t in {p13_3a, p13_3b, p13_3c} {
      \draw[thin, gray, opacity=0.4] (\s)--(\t);
    }
  }
  \foreach \s in {p24_2a, p24_2b, p24_2c} {
    \foreach \t in {p24_4a, p24_4b, p24_4c} {
      \draw[thin, gray, opacity=0.4] (\s)--(\t);
    }
  }
  \foreach \s in {p34_3a, p34_3b, p34_3c} {
    \foreach \t in {p34_4a, p34_4b, p34_4c} {
      \draw[thin, gray, opacity=0.4] (\s)--(\t);
    }
  }

  \node[draw, circle, fill=white, minimum size=0.9cm, label={[opacity=0.6, label distance=-3pt]above left:{$U_2$}}] at (0,2) {};
  \node[draw, circle, fill=white, minimum size=0.9cm, label={[opacity=0.6, label distance=-3pt]above right:{$U_4$}}] at (2,2) {};
  \node[draw, circle, fill=white, minimum size=0.9cm, label={[opacity=0.6, label distance=-3pt]below right:{$U_3$}}] at (2,0) {};
  \node[draw, circle, fill=white, minimum size=0.9cm, label={[opacity=0.6, label distance=-3pt]below left:{$U_1$}}] at (0,0) {};
  \coordinate (e1) at (-0.2, 2.2);    %
  \coordinate (e2) at (2, 2);    %
  \coordinate (e3) at (0.2, 1.8);    %
  \coordinate (e4) at (0, 0);    %
  \draw[thick, color=blue] (e1) -- (e2);
  \draw[thick, color=blue] (e2) -- (e3);
  \draw[thick, color=blue] (e3) -- (e4);
  \draw[thick, color=blue] (e4) -- (e1);
  \fill[color=blue] (e1) circle (3pt);
  \fill[color=blue] (e2) circle (3pt);
  \fill[color=blue] (e3) circle (3pt);
  \fill[color=blue] (e4) circle (3pt);
  \node[below] at (1,-0.6) {$n \binom{n}{2} \|\sv\|$};
\end{scope}

\begin{scope}[shift={(12,0)}]
  \coordinate (p12_1a) at (-0.4, 0.00);
  \coordinate (p12_1b) at ( 0.00, 0.00);
  \coordinate (p12_1c) at ( 0.4, 0.00);
  \coordinate (p12_2a) at (-0.4, 2.00);
  \coordinate (p12_2b) at ( 0.00, 2.00);
  \coordinate (p12_2c) at ( 0.4, 2.00);

  \coordinate (p13_1a) at ( 0.00,  0.4);
  \coordinate (p13_1b) at ( 0.00,  0.00);
  \coordinate (p13_1c) at ( 0.00, -0.4);
  \coordinate (p13_3a) at ( 2.00,  0.4);
  \coordinate (p13_3b) at ( 2.00,  0.00);
  \coordinate (p13_3c) at ( 2.00, -0.4);

  \coordinate (p24_2a) at ( 0.00,  2.4);
  \coordinate (p24_2b) at ( 0.00,  2.00);
  \coordinate (p24_2c) at ( 0.00,  1.6);
  \coordinate (p24_4a) at ( 2.00,  2.4);
  \coordinate (p24_4b) at ( 2.00,  2.00);
  \coordinate (p24_4c) at ( 2.00,  1.6);

  \coordinate (p34_3a) at ( 1.6,  0.00);
  \coordinate (p34_3b) at ( 2.00,  0.00);
  \coordinate (p34_3c) at ( 2.4,  0.00);
  \coordinate (p34_4a) at ( 1.6,  2.00);
  \coordinate (p34_4b) at ( 2.00,  2.00);
  \coordinate (p34_4c) at ( 2.4,  2.00);

  \foreach \s in {p12_1b} {
    \foreach \t in {p12_2a, p12_2b, p12_2c} {
      \draw[thin, gray, opacity=0.4] (\s)--(\t);
    }
  }
  \foreach \s in {p13_1b} {
    \foreach \t in {p13_3a, p13_3b, p13_3c} {
      \draw[thin, gray, opacity=0.4] (\s)--(\t);
    }
  }
  \foreach \s in {p24_2a, p24_2b, p24_2c} {
    \foreach \t in {p24_4a, p24_4b, p24_4c} {
      \draw[thin, gray, opacity=0.4] (\s)--(\t);
    }
  }
  \foreach \s in {p34_3a, p34_3b, p34_3c} {
    \foreach \t in {p34_4a, p34_4b, p34_4c} {
      \draw[thin, gray, opacity=0.4] (\s)--(\t);
    }
  }

  \node[draw, circle, fill=white, minimum size=0.9cm, label={[opacity=0.6, label distance=-3pt]above left:{$U_2$}}] at (0,2) {};
  \node[draw, circle, fill=white, minimum size=0.9cm, label={[opacity=0.6, label distance=-3pt]above right:{$U_4$}}] at (2,2) {};
  \node[draw, circle, fill=white, minimum size=0.9cm, label={[opacity=0.6, label distance=-3pt]below right:{$U_3$}}] at (2,0) {};
  \node[draw, circle, fill=white, minimum size=0.9cm, label={[opacity=0.6, label distance=-3pt]below left:{$U_1$}}] at (0,0) {};
  \coordinate (e1) at (0, 2);    %
  \coordinate (e2) at (-0.2, -0.2);    %
  \coordinate (e3) at (2, 0);    %
  \coordinate (e4) at (0.2, 0.2);    %
  \draw[thick, color=blue] (e1) -- (e2);
  \draw[thick, color=blue] (e2) -- (e3);
  \draw[thick, color=blue] (e3) -- (e4);
  \draw[thick, color=blue] (e4) -- (e1);
  \fill[color=blue] (e1) circle (3pt);
  \fill[color=blue] (e2) circle (3pt);
  \fill[color=blue] (e3) circle (3pt);
  \fill[color=blue] (e4) circle (3pt);
  \node[below] at (1,-0.6) {$n^2 \binom{\inn{\sv, \iq}}{2}$};
\end{scope}

\begin{scope}[shift={(18,0)}]
  \coordinate (p12_1a) at (-0.4, 0.00);
  \coordinate (p12_1b) at ( 0.00, 0.00);
  \coordinate (p12_1c) at ( 0.4, 0.00);
  \coordinate (p12_2a) at (-0.4, 2.00);
  \coordinate (p12_2b) at ( 0.00, 2.00);
  \coordinate (p12_2c) at ( 0.4, 2.00);

  \coordinate (p13_1a) at ( 0.00,  0.4);
  \coordinate (p13_1b) at ( 0.00,  0.00);
  \coordinate (p13_1c) at ( 0.00, -0.4);
  \coordinate (p13_3a) at ( 2.00,  0.4);
  \coordinate (p13_3b) at ( 2.00,  0.00);
  \coordinate (p13_3c) at ( 2.00, -0.4);

  \coordinate (p24_2a) at ( 0.00,  2.4);
  \coordinate (p24_2b) at ( 0.00,  2.00);
  \coordinate (p24_2c) at ( 0.00,  1.6);
  \coordinate (p24_4a) at ( 2.00,  2.4);
  \coordinate (p24_4b) at ( 2.00,  2.00);
  \coordinate (p24_4c) at ( 2.00,  1.6);

  \coordinate (p34_3a) at ( 1.6,  0.00);
  \coordinate (p34_3b) at ( 2.00,  0.00);
  \coordinate (p34_3c) at ( 2.4,  0.00);
  \coordinate (p34_4a) at ( 1.6,  2.00);
  \coordinate (p34_4b) at ( 2.00,  2.00);
  \coordinate (p34_4c) at ( 2.4,  2.00);

  \foreach \s in {p12_1b} {
    \foreach \t in {p12_2a, p12_2b, p12_2c} {
      \draw[thin, gray, opacity=0.4] (\s)--(\t);
    }
  }
  \foreach \s in {p13_1b} {
    \foreach \t in {p13_3a, p13_3b, p13_3c} {
      \draw[thin, gray, opacity=0.4] (\s)--(\t);
    }
  }
  \foreach \s in {p24_2a, p24_2b, p24_2c} {
    \foreach \t in {p24_4a, p24_4b, p24_4c} {
      \draw[thin, gray, opacity=0.4] (\s)--(\t);
    }
  }
  \foreach \s in {p34_3a, p34_3b, p34_3c} {
    \foreach \t in {p34_4a, p34_4b, p34_4c} {
      \draw[thin, gray, opacity=0.4] (\s)--(\t);
    }
  }

  \node[draw, circle, fill=white, minimum size=0.9cm, label={[opacity=0.6, label distance=-3pt]above left:{$U_2$}}] at (0,2) {};
  \node[draw, circle, fill=white, minimum size=0.9cm, label={[opacity=0.6, label distance=-3pt]above right:{$U_4$}}] at (2,2) {};
  \node[draw, circle, fill=white, minimum size=0.9cm, label={[opacity=0.6, label distance=-3pt]below right:{$U_3$}}] at (2,0) {};
  \node[draw, circle, fill=white, minimum size=0.9cm, label={[opacity=0.6, label distance=-3pt]below left:{$U_1$}}] at (0,0) {};
  \coordinate (e1) at (0, 0.2);    %
  \coordinate (e2) at (2, 0.2);    %
  \coordinate (e3) at (0, -0.2);    %
  \coordinate (e4) at (2, -0.2);    %
  \draw[thick, color=blue] (e1) -- (e2);
  \draw[thick, color=blue] (e2) -- (e3);
  \draw[thick, color=blue] (e3) -- (e4);
  \draw[thick, color=blue] (e4) -- (e1);
  \fill[color=blue] (e1) circle (3pt);
  \fill[color=blue] (e2) circle (3pt);
  \fill[color=blue] (e3) circle (3pt);
  \fill[color=blue] (e4) circle (3pt);
  \node[below] at (1,-0.6) {$\binom{n}{2} \binom{\|\iq\|}{2}$};
\end{scope}

\end{tikzpicture}
\caption{Four %
types of copies of $H = C_4$ 
in the query graph, %
along with their counts. %
}
\label{fig:c4-querygraph}
\end{figure}

%% file: algs/lndp-counting.tex
\begin{algorithm}[ht!]
    \caption{$\algsubh$ for counting occurrences of template graph $H$.}
    \label{alg:subh}
    
    \begin{algorithmic}[1]
        \Statex \textbf{Parameters:} Privacy parameter $\eps \in (0,1]$;
        $k,n\in\N$; template graph $H$ on node set $[k]$.
        \Statex \textbf{Input:}
        Graph $G$ on vertex set $[n]$.
        \Statex \textbf{Output:} Subgraph count estimate $\wh T$.
        \State Let $\eps' = \eps /
        \vc{H}
        $, and let $\msf{C}\subseteq[k]$ be
        a minimum-size vertex cover of $H$.
        \State Use public randomness to select a uniformly random permutation $P$ of $[n]$, divided into sequences $P^{(1)},\ldots, P^{(n/k)}$ of size $k$.
        \label{line:pub-rand-alg}

        \medskip
        \Statex \textit{Local randomization:}
        \smallskip
        \For{{\bf each} sequence $P^{(j)} = (v^{(j)}_1, \ldots, v^{(j)}_k)$}
            \For{{\bf each} $i\in \msf{C}$} \label{line:each-node-loop}
                \State Node $v_i^{(j)}$ releases $Y_i^{(j)} \gets \frac{X_i^{(j)}\cdot (e^{\eps'} + 1) - 1}{e^{\eps'} - 1}$,
                where $X_i^{(j)} \gets \rr[\eps']\bparen{\edgematch^{P^{(j)},H}_i}$.
                \label{line:sub-node-release}
                \label{line:sub-debias}
                \Statex \Comment{\commentstyle{$Y_i^{(j)}$ is a debiased version of $X_i^{(j)}$.}}
            \EndFor
        \EndFor

    \medskip
    \Statex \textit{Aggregation (post-processing):}

    \smallskip
    \State For each sequence $P^{(j)}$, set  $\wh S^{(j)} \gets \prod_{i\in \msf{C}} Y_i^{(j)}$.
    
    \smallskip
    \State \Return $\wh T \gets \frac{k!}{|\Aut(H)|} \cdot \binom{n-1}{k-1} \cdot \sum_{j\in[n/k]} \wh S^{(j)} .$
    \label{line:subh-ans}
    \end{algorithmic}
\end{algorithm}

%% file: sections/lower-bound-polynomial.tex
\section{Error Required for
Subgraph Counting
with \texorpdfstring{\LNDP}{LNDP}}
\label{sec:lb}

In this section
we prove \Cref{thm:lb}, showing
a lower bound of $\Omega_k(n^{k-\frac{1}{2}})$ on the additive error required by \LNDP algorithms for counting $k$-node subgraphs $H$ with maximum degree
greater than $1$.
We %
describe %
inner product queries and the properties of our main reduction
in \Cref{sec:ip-back}, %
our method for encoding an inner product query on a secret dataset
as a \emph{query graph} in \Cref{sec:q-graph},
and provide our main reduction from inner-product queries in the central model to $H$-counting under \LNDP in \Cref{sec:lb-pf}.

\subsection{Inner Product Queries}
\label{sec:ip-back}

An inner product query  on dataset $\sv \in\zo^n$ is
specified
by a vector $\iq \in \{0, 1\}^n$, and maps $\sv$ to the answer $\ia = \inn{\sv, \iq}$. The problem $\innerproducts(n, q)$ is to answer $q$
inner product queries $\iq^{(1)}, \ldots, \iq^{(q)} \in \{0, 1\}^n$ on dataset $\sv\in\zo^n$.
We consider algorithms that answer inner product queries with the guarantee that, with high probability,
all but a $\gamma$-fraction of their answers are accurate.

\begin{definition}[Inner product accuracy]
\label{def:inn-prod-acc-tri}
    Let $\alpha, \gamma > 0$ and $n, q \in \mathbb{N}$. A randomized algorithm $\cB$ is 
    \emph{$(\alpha, \gamma)$-accurate} for $\innerproducts(n, q)$ if, for all queries $\iq^{(1)}, \ldots, \iq^{(q)} \in \{0, 1\}^n$ and all datasets $\sv \in \{0, 1\}^n$, algorithm 
    $\cB\bparen{\sv,\bparen{\iq^{(j)}}_{j\in [q]}}$
    outputs answers $\ia^{(1)}, \ldots, \ia^{(q)} \in \R$ such that
    \[
        \Pr_{\text{coins of $\cB$}}
        \Big[\Big|\set{i \in [q] : \left|\ia^{(i)} - \inn{\iq^{(i)}, \sv} \right| > \alpha}\Big|
        \leq
        \gamma \cdot q \Big]
        \geq 
        \frac{99}{100},
    \]
    i.e., with probability at least
    $\frac{99}{100}$,
    all but a $\gamma$-fraction of answers $\ia^{(i)}$
    satisfy $\left|\ia^{(i)} - \inn{\innquery^{(i)}, \sv} \right| \leq \alpha$.
\end{definition}

An algorithm for $\innerproducts$ is $(\eps,\delta)$-DP
if, for every fixed set of query vectors and all vectors $\sv$ and $\sv'$ that differ in one bit, %
the output distributions of the algorithm are $(\eps,\delta)$-indistinguishable.
A slight adaptation of the reconstruction attacks of \cite{DinurN03}, as reanalyzed in \cite{DworkMT07}, gives the following impossibility result (see \Cref{sec:lem-inn-prod-explanation} for a proof sketch).

\begin{restatable}[Inner product lower bound, based on {\cite{DinurN03}, \cite[Theorem 23]{DworkMT07}}]{lemma}{InnerProductDPLemma}
\label{lem:inn-prod-dp-lb-new}
    There exist $c_1,c_2 > 0$ such that, for sufficiently large $n \in \mathbb{N}$, if a randomized algorithm $\cB$ is $(c_1 \sqrt{n}, \frac{1}{100} )$-accurate for $\innerproducts(n, c_2 n)$, then $\cB$ is not $(4, \frac{1}{25})$-DP.
\end{restatable}

\paragraph{Key reduction.} %
We prove our lower bound by reducing from inner product queries in the central model to accurate subgraph counting under \LNDP. %

\begin{lemma}[Inner products via subgraph counts]
\label{lem:inn-prod-reduction-new}
Let
$\alpha > 0$ and $k, n, q \in \mathbb N$ with $n \geq 5k^k$.
Let $H$ be a $k$-node graph with maximum degree greater than $1$.
Suppose there is a $(1, \frac{1}{200})$-\LNDP algorithm $\cA$ that, for every $2kn$-node graph $G$, counts the number of copies of $H$ in $G$ with error
\begin{equation}
\label{eqn:lem-inn-prod-red-err}
\Ex_{\text{coins of $\cA$}}[|\cA(G) - \subgraphcount(G)|] \leq 
(2n)^{k-1} \cdot 9 \cdot 10^{-6} \cdot \alpha = O_k(\alpha \cdot n^{k-1}).
\end{equation}
Then there is a central-model $(4, \frac{1}{25})$-DP algorithm $\cB$ that is $(\alpha, \frac{1}{100})$-accurate for $\innerproducts(n, q)$.
\end{lemma}

\Cref{lem:inn-prod-dp-lb-new,lem:inn-prod-reduction-new} immediately allow us to prove \Cref{thm:lb}.

\begin{proof}[Proof of \Cref{thm:lb}]
Let $c_1$ and $c_2$ be the constants from \Cref{lem:inn-prod-dp-lb-new} and let
$n' \in \mathbb{N}$ sufficiently large with $n' \geq 5k^k$.
Let $n = 2kn'$, and suppose that $\cA$ is a $(1, \frac{1}{200})$-\LNDP algorithm that, for graphs $G$ on $2kn'$ nodes, counts copies of $H$ with expected error $\Ex[|\cA(G) - \subgraphcount(G)|] = \alpha$, where $\alpha \leq (2n')^{k-1} \cdot 9 \cdot 10^{-6} (c_1 \sqrt{n'})$.
Applying \Cref{lem:inn-prod-reduction-new} gives an algorithm $\cB$ that is $(4, \frac{1}{25})$-DP and $(c_1 \sqrt{n'}, \frac{1}{100})$-accurate for $\innerproducts(n', c_2 n')$. However, by \Cref{lem:inn-prod-dp-lb-new}, $\cB$ cannot be $(4, \frac{1}{25})$-DP, which is a contradiction. 
Thus, we must have $\alpha > (2n')^{k-1} \cdot 9 \cdot 10^{-6} (c_1 \sqrt{n'}) = \Omega_k(n^{k - \frac{1}{2}})$, as desired.
\end{proof}

The rest of this section is devoted to proving \Cref{lem:inn-prod-reduction-new}.

\subsection{Query Graphs and Their Properties}
\label{sec:q-graph}

We now describe how to encode a secret dataset and inner product query as a \emph{query graph}.
Let $H = ([k], E(H))$ be an undirected graph on node set $[k]$ with maximum degree
greater than $1$.
W.l.o.g.\  assume $\{\{1, 2\}, \{1, 3\}\} \subseteq E(H)$, i.e., node $1$ is adjacent  
to nodes $2$ and $3$, which can be achieved by relabeling nodes since $\dmax(H) > 1$.  A query graph when $H$ is a triangle is depicted in \Cref{fig:querygraph}.

\begin{definition}[Query graph] \label{def:query-graph-new} Fix a secret dataset $\sv=(x_1,\ldots,x_n)\in\{0,1\}^n$ and an inner product query $\iq=(a_1,\ldots,a_n)\in\{0,1\}^n$. The \emph{query graph} $G[\sv,\iq]$ has $N:=kn$ nodes partitioned into sets $U_1,\ldots,U_k$ of size $n$, with nodes in $U_s$ denoted $\node{s}{1},\ldots,\node{s}{n}$. We add edges only between parts $U_s,U_t$ with $\{s,t\}\in E(H)$, as follows: \begin{itemize} 
\item For each $i\in[n]$, add edges from $\node{1}{i}$ to all of $U_2$ iff $x_i=1$, and to all of $U_3$ iff $a_i=1$. 
\item For every other pair $\{s,t\}\in E(H)$, form a complete bipartite graph on parts $U_s$ and $U_t$. \end{itemize} 
\end{definition}

As outlined in \Cref{sec:intro-lb-tech},
the query graph has several important properties (slightly abbreviated here):

\begin{enumerate}
    \item[R1:] 
    A single bit change in $\sv$ corresponds to a rewiring of one node in $U_1$.
    \item[R2:] 
    An error of $\alpha$ in the subgraph count $\subgraphcount(G[\sv, \iq])$ translates to an error of $O(\alpha / n^{k-1})$ in the estimate for the inner product query $\inn{\sv, \iq}$.
    \item[R3:] For every fixed $\sv$, each node in $U_1$ and $U_2$
    has \emph{two} possible neighborhoods 
    in the query graphs $\bset{G[\sv,\iq] : \iq \in \set{0,1}^{n}}$, namely its neighborhoods in  $G[\sv,\vec 0]$ and $G[\sv,\vec 1]$.
\end{enumerate}

While R1 and R3 are immediate, showing R2 requires more work, which we now describe.

To understand the number of copies of $H$ in the query graph, we first classify the vertices in $U_1$ according to how they connect to $U_2$ and $U_3$.
Partition
$U_1$ into %
subsets $U_1^{11}, U_1^{10}, U_1^{01}, U_1^{00}$ denoting the nodes $\node{1}{i}$ with $(x_i, a_i)$ equal to $(1, 1), (1, 0), (0, 1)$ and $(0, 0)$, respectively. 
The nodes in each subset 
have the same neighborhoods in $U_2\cup U_3$: nodes in $U_1^{11}$ connect to both $U_2$ and $U_3$, nodes in $U_1^{10}$ connect to $U_2$ but not $U_3$, nodes in $U_1^{01}$ connect to $U_3$ but not $U_2$, and nodes in $U_1^{00}$ connect to neither. Their sizes are
\begin{equation}
\label{eqn:lb-subpart-sizes}
\begin{array}{ll@{\hspace{2cm}}l}
    |U_1^{11}| = \inn{\sv, \iq} &&
    |U_1^{10}| = \|\sv\| - \inn{\sv, \iq} \\
    |U_1^{01}| = \|\iq\| - \inn{\sv, \iq} &&
    |U_1^{00}| = n - \|\sv\| - \|\iq\| + \inn{\sv, \iq}.
\end{array}
\end{equation}

Since all nodes in a fixed $U_s$ for $s \geq 2$ have identical neighborhoods and so do the nodes in a fixed part of $U_1$ (using the partition above), the number of copies of $H$ in the query graph $G[\sv, \iq]$ depends solely on $n$ and the sizes of the subsets in \Cref{eqn:lb-subpart-sizes}, 
which in turn depend only on %
$p=\inn{\sv,\iq}, w_x=\|\sv\|$, and $w_a=\|\iq\|$.
This motivates the %
\emph{$H$-counting function}, which lets us study how the subgraph count changes with $\inn{\sv, \iq}$
while the weights of $\sv$ and $\iq$ stay fixed.

\begin{definition}[$H$-counting function]
Fix $\sv = (x_1, \ldots, x_n) \in \{0, 1\}^n$ and $\iq = (a_1, \ldots, a_n) \in \{0, 1\}^n$. Define the \emph{$H$-counting function} $\f(p,w_x,w_a):=\subgraphcount(G[\sv, \iq])$ to be the number of copies of $H$ in the query graph for any $\sv$ and $\iq$ such that $p = \inn{\sv, \iq}$, $w_x = \|\sv\|$, and $w_a = \|\iq\|$.
Furthermore, define the \emph{forward difference}
\[
    \pdv{\f}{p}(p, w_x, w_a)
    :=
    \f(p+1, w_x, w_a) - \f(p, w_x, w_a).
\]
\end{definition}

In the following lemma we lower bound the forward difference of $\f$ with respect to $p$. %
As we show subsequently, this bound ensures that even an approximate version of $\f\bparen{\inn{\sv, \iq}, \|\sv\|, \|\iq\|}$ reveals the approximate value of $\inn{\sv, \iq}$
when $\|\sv\|$ and $\|\iq\|$ are publicly known.

We say a triple $(p, w_x, w_a) \in \{0, \ldots, n\}^3$ is \emph{feasible} if there exist $\sv, \iq \in \{0, 1\}^n$ such that $p = \inn{\sv, \iq}, w_x = \|\sv\|$, and $w_a = \|\iq\|$ (i.e., if $w_x+w_a-n \leq p \leq \min\{w_x,w_a\}$).

\begin{lemma}[Forward difference bound]
\label{lem:f-properties}
For every feasible $(p,w_x,w_a)$, if $(p+1, w_x, w_a)$ is also feasible and $n \geq 10k^k$, then $\pdv{\f}{p}(p, w_x, w_a) \geq \frac{9}{10} n^{k-1}$.
\end{lemma}
\begin{proof}[Proof of \Cref{lem:f-properties}]

Let $\sv$ and $\iq$ be such that $p = \inn{\sv, \iq}$, $w_x = \|\sv\|$, and $w_a = \|\iq\|$, and let $G := G[\sv, \iq]$. Since $(p+1, w_x, w_a)$ is feasible, there exist $u^+ := \node{1}{i} \in U_1^{10}$ and $u^{-} := \node{1}{j} \in U_1^{01}$ for some $i, j \in [n]$. Let $G' := G[\sv', \iq']$, where $\iq'$ agrees with $\iq$ except that $a_{i}' = 1$ and $a_{j}' = 0$. Then $\|\iq'\| = \|\iq\|$ and $\inn{\sv, \iq'} = p+1$, so $\subgraphcount(G') = \f(p+1, w_x, w_a)$.%

Let $A$ be the set of copies in $G'$ but not $G$, and let $B$ be the set of 
copies in $G$ but not $G'$. By definition of $G$ and $G'$, every copy in $A$ uses an edge between $u^+$ and $U_3$, and every copy in $B$ uses an edge between $u^-$ and $U_3$. Write the disjoint union $B = B_1 \sqcup B_2$, where $B_1$ is the set of copies in $B$ that contain $u^+$. Then,
\begin{equation*}
\Delta_p \f(p,w_x,w_a) = \subgraphcount(G') - \subgraphcount(G) = |A|-|B| = (|A| - |B_2|) - |B_1|.
\end{equation*}
It suffices to show $|B_1| \leq k^k n^{k-2}$ and $|A|-|B_2| \geq n^{k-1}$, which, combined with $n \geq 10k^k$, proves the lemma.

\textit{Bounding $|B_1|$.}
There are $k$ choices for the nodes in $H$ assigned to each of
$u^-$ and $u^+$ and $(kn)^{k-2}$ choices for 
the remaining $k-2$ nodes in $H$, giving $|B_1| \leq k^2(kn)^{k-2} \leq k^k n^{k-2}$.

\textit{Bounding $|A| - |B_2|$.}
For $c \in B_2$, let $\psi(c)$ be the copy obtained by replacing $u^-$ with $u^+$. 
This gives a copy in $G'$ that is not in $G$, so $\psi$ maps $B_2$ to $A$ and is injective (one can replace $u^+$ back with $u^-$), so $|\psi(B_2)| = |B_2|$. 
There are $n^{k-1}$ copies of $H$ in $A$ %
mapping node $1$ of $H$ to $u^+$ and each node $s \geq 2$ to $U_s$. Each such copy is not in $\psi(B_2)$, as it has an edge between $u^+$ and $U_2$, whereas $u^-$ has no neighbors in $U_2$. This gives $|A| - |B_2| = |A \setminus \psi(B_2)| \geq n^{k-1}$, as desired.
\end{proof}

\subsection{Reduction from Inner Product Queries to Subgraph Counting}
\label{sec:lb-pf}

In this section, we describe and analyze the reduction from inner product queries to subgraph counting, proving \Cref{lem:inn-prod-reduction-new}. 
\Cref{fig:reduction-outline} illustrates the structure of
the reduction in \Cref{alg:inn-prod-reduc-new}, which was described informally in the introduction.
Rather than directly working with the secret dataset $\sv \in \{0, 1\}^n$ and inner product query $\iq^{(\ell)} \in \{0, 1\}^n$, our reduction applies the query graph construction and the $H$-counting function to the concatenations $\svc := \sv \circ (\vec 1 - \sv)$ and $\iqc^{(\ell)} := \iq^{(\ell)} \circ \vec 0$ of length $2n$.
Note that $\|\svc\| = n$ and $\|\iqc^{(\ell)}\| = \|\iq^{(\ell)}\|$ are publicly known (although $\sv$ is private), and $\langle \svc, \iqc^{(\ell)} \rangle = \langle \sv, \iq^{(\ell)} \rangle$. This results in just one private unknown $\langle \svc, \iqc^{(\ell)} \rangle$, which our reduction estimates.

The algorithm accesses the secret vector $\sv$ by making two calls to the randomizers of the \LNDP algorithm $\cA$---one on $G[\svc,\vec 0]$ and one on $G[\svc,\vec 1]$---to obtain the necessary views of nodes to reconstruct an output of $\cA$ on every query graph.

To answer each inner product query $\iq^{(\ell)}$, the reduction %
uses the MixAndMatch subroutine (\Cref{alg:mix-and-match}), which postprocesses the obtained randomizer reports to obtain an estimate $d^{(\ell)} \sim \cA(G[\svc, \iqc^{(\ell)}])$ of the subgraph count on $G[\svc, \iqc^{(\ell)}]$ %
without actually making another call to $\cA$.
Then, using the publicly known weights $\|\svc\| = n$ and $\|\iqc^{(\ell)}\|$, the reduction constructs the function $g_\ell(\cdot) = f_H(\cdot, \|\svc\|, \|\iqc^{(\ell)}\|) = f_H(\cdot, n, \|\iqc^{(\ell)}\|)$, which depends only on the inner product $\langle \svc, \iqc^{(\ell)} \rangle = \inn{\sv, \iq^{(\ell)}}$.
Since $d^{(\ell)}$ is an estimate of the count of copies of $H$, 
the algorithm finds the value $\wh p^{(\ell)}$ minimizing %
$|g_\ell(\wh p^{(\ell)}) - d^{(\ell)}|$, which it %
outputs as an estimate for the inner product. %

\input{figures/reduction-diagram}

We first show that each call to the mix-and-match subroutine produces output that faithfully simulates a call to $\cA$ on the appropriate query graph $G[\svc, \iqc]$ (\Cref{lem:mix-and-match}). We then use this result to analyze the full reduction, in the proof of \Cref{lem:inn-prod-reduction-new}.

\begin{lemma}[Mix--and--match]
\label{lem:mix-and-match}
Consider \Cref{alg:mix-and-match} with parameters $n, H, \cA = (\vec \cR, \cP), \rho$, and inputs $\iq$ and $\nodeview$. 
The returned value $d$ satisfies
 $d \sim \cA(G[\sv, \iq])$.
\end{lemma}

\input{algs/alg-reduction-polynomial}

\input{algs/mix-and-match-polynomial}

\begin{proof}[Proof of \Cref{lem:mix-and-match}]
By \Cref{def:lndp}, it
suffices to show that, for all nodes $\node{s}{i} \in [kn]$, we have $r(\node{s}{i}) \sim \cR_{\node{s}{i}, \rho}(N_{\node{s}{i}}(G[\sv, \iq]))$, which we show by considering each part $U_s$.

\emph{For $s \in \{1, 2\}$:}
Note that $(a_i \cdot \vec 1)_i = a_i$, so by \Cref{def:query-graph-new}, the neighborhood of node $\node{s}{i}$ is equal in $G_1 := G[\sv, a_i \cdot \vec 1]$ and $G_2 := G[\sv, \iq]$. %
Thus, due to Line~\ref{line:lb-set-ui-new} of \Cref{alg:mix-and-match} and Line~\ref{line:construct-vec-r-new} of \Cref{alg:inn-prod-reduc-new}, %
$$
r(\node{s}{i})
= \nodeview[\node{s}{i}, a_i \cdot \vec 1]
\sim \cR_{\node{s}{i}, \rho} (N_{\node{s}{i}}(G_1))
= \cR_{\node{s}{i}, \rho} (N_{\node{s}{i}}(G_2)).
$$

\emph{For %
$s \in \{3, \ldots, k\}$:}
By construction, for $s \in \{3, \ldots, k\}$, the edge list $N_{\node{s}{i}}(G[\sv, \iq])$ is independent of $\sv$. Therefore, by Line~\ref{line:lb-set-w-new} of \Cref{alg:mix-and-match}, we have
\[r(\node{s}{i}) 
\sim 
\cR_{\node{s}{i}, \rho}(N_{\node{s}{i}}(G[\vec 0, \iq]))
=
\cR_{\node{s}{i}, \rho}(N_{\node{s}{i}}(G[\sv, \iq])) \, . \qedhere\]
\end{proof}

\bigskip
We now complete the analysis of the reduction by proving \Cref{lem:inn-prod-reduction-new}.

\begin{proof}[Proof of \Cref{lem:inn-prod-reduction-new}]

Let $\cA$ be a $(1, \frac{1}{200})$-\LNDP algorithm that achieves the stated error. By \Cref{def:lndp}, it is specified by a distribution $\Phi$ over public randomness strings $\rho$, randomizers $\vec \cR_\pubrand$, %
and postprocessing algorithm $\cP$.
We separately show that algorithm $\cB$ given by \Cref{alg:inn-prod-reduc-new} is $(4, \frac{1}{25})$-DP, and that it is $(\alpha, \frac{1}{100})$-accurate for $\innerproducts(n, q)$.

\emph{(Privacy.)}
The secret vector $\sv$ is accessed on Line~\ref{line:construct-vec-r-new} of \Cref{alg:inn-prod-reduc-new}, where it is queried a total of $|\{\vec 0, \vec 1\}| = 2$ times.
Changing one bit of $\sv$ changes two bits of $\svc$, so for neighboring datasets $\sv, \sv'$, the corresponding query graphs $G[\svc, \vec c], G[\svc', \vec c]$ are at node distance $2$.
Since $\cA$ is $(1, \frac{1}{200})$-\LNDP, applying its node-privacy guarantee twice shows that each output of $\vec \cR_\pubrand(G[\svc, \vec c])$ is $(2, \frac{1+e}{200})$-DP.
By composition (\Cref{lem:basic-comp}) and since $2 \cdot \frac{1+e}{200} \leq \frac{1}{25}$, \Cref{alg:inn-prod-reduc-new} satisfies $(4, \frac{1}{25})$-DP.

\emph{(Accuracy.)}
The following claim establishes the error of a single estimate $\wh p^{(\ell)}$:
\begin{claim}
\label{clm:query-acc}
Fix $\iq \in \!\{\iq^{(1)}\!, \ldots, \iq^{(q)}\}$, and 
let $d \!\sim \!\cA(G[\svc, \iqc])$. As on Line~\ref{line:def-p} of \Cref{alg:inn-prod-reduc-new}, set $\wh p = \arg \min_{p \in \{0, \ldots, \|\iq\|\}} |g(p) - d|$.
Then $\E[|\wh p - \inn{\sv, \iq}\!|] \leq 2 \cdot 10^{-5} \cdot \alpha$.
\end{claim}

\begin{proof}[Proof of \Cref{clm:query-acc}]
Let $p = \inn{\sv, \iq}$ and $w_a = \|\iqc\| = \|\iq\|$.
Since $\|\svc\| = n$ and $\inn{\svc, \iqc} = \inn{\sv, \iq}$, we have $g(p) := \f(p, n, w_a) = \subgraphcount(G[\svc, \iqc])$. Let $e_H := d - g(p)$ be the additive error of $\cA$ on $G[\svc, \iqc]$. By summing the forward differences of $\f$ using \Cref{lem:f-properties} (with $2n \geq 10k^k$) and using the definition of $\hat p$, we have
$$
\frac{9}{10} (2n)^{k-1} \cdot |\hat p - p| \leq |g(\hat p) - g(p)| \leq |g(\hat p) - d| + |g(p) - d| \leq 2 |g(p) - d| \leq 2|e_H|.
$$
Rearranging, taking expectations, and using the assumption on $\E[|e_H|]$ in \Cref{eqn:lem-inn-prod-red-err} gives
\[
\E[|\hat p - p|] \leq \frac{20}{9 \cdot (2n)^{k-1}} \E[|e_H|] \leq \frac{20}{9 \cdot (2n)^{k-1}} \cdot (2n)^{k-1} \cdot 9 \cdot 10^{-6} \cdot \alpha \leq 2 \cdot 10^{-5} \cdot \alpha. \qedhere
\]

\end{proof}

To finish the proof of \Cref{lem:inn-prod-reduction-new},
let $X_\ell = |\wh p^{(\ell)} - \inn{\sv, \iq^{(\ell)}}|$ for all $\ell \in [q]$, and 
define $Y = |\{\ell \in [q] : X_\ell > \alpha\}|$. Using \Cref{clm:query-acc} and applying Markov's inequality to each $X_\ell$ and $Y$ gives $\Pr[X_\ell > \alpha] \leq 2 \cdot 10^{-5}$ for all $\ell \in [q]$, so that
\begin{align*}
    \Pr \left[Y > \frac{q}{100} \right] \leq \frac{2 \cdot 10^{-5} \cdot q}{q/100} = 2 \cdot 10^{-3} < \frac{1}{100}.
\end{align*}
That is, with probability at least $\frac{99}{100}$, we have
$\left| \{\ell \in [q] : \left|\wh p^{(\ell)} - \inn{\iq^{(\ell)}, \sv}\right| > \alpha\} \right| \leq \frac{q}{100}$,
so $\cB$ is $(\alpha, \frac{1}{100})$-accurate for $\innerproducts(n, q)$, completing the proof.
\end{proof}

%% file: figures/reduction-diagram.tex
\begin{figure}[ht!]
\centering

\tikzset{every picture/.style={line width=0.75pt}} %

\begin{tikzpicture}[x=0.75pt,y=0.75pt,yscale=-1,xscale=1]

\draw  [color={rgb, 255:red, 0; green, 0; blue, 0 }  ,draw opacity=1 ][fill={rgb, 255:red, 255; green, 255; blue, 255 }  ,fill opacity=1 ] (360,113) -- (449.67,113) -- (449.67,229.8) -- (360,229.8) -- cycle ;
\draw  [color={rgb, 255:red, 0; green, 0; blue, 0 }  ,draw opacity=1 ][fill={rgb, 255:red, 255; green, 255; blue, 255 }  ,fill opacity=1 ] (297.17,136) -- (386.83,136) -- (386.83,252) -- (297.17,252) -- cycle ;
\draw    (93.4,125.48) -- (93.4,216.2) ;
\draw    (93.4,125.6) -- (110.4,125.6) ;
\draw [shift={(113.4,125.6)}, rotate = 180] [fill={rgb, 255:red, 0; green, 0; blue, 0 }  ][line width=0.08]  [draw opacity=0] (8.04,-3.86) -- (0,0) -- (8.04,3.86) -- (5.34,0) -- cycle    ;
\draw    (80.36,173) -- (93.2,173) ;
\draw    (93.8,216.2) -- (110.8,216.2) ;
\draw [shift={(113.8,216.2)}, rotate = 180] [fill={rgb, 255:red, 0; green, 0; blue, 0 }  ][line width=0.08]  [draw opacity=0] (8.04,-3.86) -- (0,0) -- (8.04,3.86) -- (5.34,0) -- cycle    ;
\draw    (166,126.2) -- (187.8,126.2) ;
\draw [shift={(190.8,126.2)}, rotate = 180] [fill={rgb, 255:red, 0; green, 0; blue, 0 }  ][line width=0.08]  [draw opacity=0] (8.04,-3.86) -- (0,0) -- (8.04,3.86) -- (5.34,0) -- cycle    ;
\draw    (167,216.2) -- (188.8,216.2) ;
\draw [shift={(191.8,216.2)}, rotate = 180] [fill={rgb, 255:red, 0; green, 0; blue, 0 }  ][line width=0.08]  [draw opacity=0] (8.04,-3.86) -- (0,0) -- (8.04,3.86) -- (5.34,0) -- cycle    ;
\draw   (193.6,81.6) -- (233.96,81.6) -- (233.96,161.6) -- (193.6,161.6) -- cycle ;
\draw   (193.6,176.6) -- (233.96,176.6) -- (233.96,256.6) -- (193.6,256.6) -- cycle ;
\draw  [color={rgb, 255:red, 0; green, 0; blue, 0 }  ,draw opacity=1 ] (45.56,63.5) -- (456.83,63.5) -- (456.83,275.08) -- (45.56,275.08) -- cycle ;
\draw  [color={rgb, 255:red, 0; green, 0; blue, 0 }  ,draw opacity=1 ][fill={rgb, 255:red, 255; green, 255; blue, 255 }  ,fill opacity=1 ] (270.33,149) -- (360,149) -- (360,266) -- (270.33,266) -- cycle ;
\draw  [color={rgb, 255:red, 0; green, 0; blue, 0 }  ,draw opacity=1 ] (273.83,153.82) -- (356.17,153.82) -- (356.17,185.8) -- (273.83,185.8) -- cycle ;
\draw    (315,186) -- (315,200.5) ;
\draw [shift={(315,203.5)}, rotate = 270] [fill={rgb, 255:red, 0; green, 0; blue, 0 }  ][line width=0.08]  [draw opacity=0] (8.04,-3.86) -- (0,0) -- (8.04,3.86) -- (5.34,0) -- cycle    ;
\draw    (315,227.6) -- (326.69,237.29) ;
\draw [shift={(329,239.2)}, rotate = 219.64] [fill={rgb, 255:red, 0; green, 0; blue, 0 }  ][line width=0.08]  [draw opacity=0] (8.04,-3.86) -- (0,0) -- (8.04,3.86) -- (5.34,0) -- cycle    ;
\draw    (233.96,216.6) .. controls (260.21,216.18) and (243.06,179.01) .. (270.49,176.78) ;
\draw [shift={(273.16,176.68)}, rotate = 180] [fill={rgb, 255:red, 0; green, 0; blue, 0 }  ][line width=0.08]  [draw opacity=0] (8.04,-3.86) -- (0,0) -- (8.04,3.86) -- (5.34,0) -- cycle    ;
\draw    (315.04,93.27) -- (315.16,146) ;
\draw [shift={(315.17,149)}, rotate = 269.87] [fill={rgb, 255:red, 0; green, 0; blue, 0 }  ][line width=0.08]  [draw opacity=0] (8.04,-3.86) -- (0,0) -- (8.04,3.86) -- (5.34,0) -- cycle    ;
\draw    (373.56,252.28) .. controls (373.17,271.6) and (352.9,269.39) .. (351.48,294.51) ;
\draw [shift={(351.4,297.35)}, rotate = 270] [fill={rgb, 255:red, 0; green, 0; blue, 0 }  ][line width=0.08]  [draw opacity=0] (8.04,-3.86) -- (0,0) -- (8.04,3.86) -- (5.34,0) -- cycle    ;
\draw    (338.76,265.88) .. controls (338.44,286.9) and (316.84,272.54) .. (315.11,294.43) ;
\draw [shift={(315,297.35)}, rotate = 270] [fill={rgb, 255:red, 0; green, 0; blue, 0 }  ][line width=0.08]  [draw opacity=0] (8.04,-3.86) -- (0,0) -- (8.04,3.86) -- (5.34,0) -- cycle    ;
\draw    (350.76,92.88) .. controls (350.43,127) and (342.42,95.97) .. (342.02,133.01) ;
\draw [shift={(342,136)}, rotate = 270] [fill={rgb, 255:red, 0; green, 0; blue, 0 }  ][line width=0.08]  [draw opacity=0] (8.04,-3.86) -- (0,0) -- (8.04,3.86) -- (5.34,0) -- cycle    ;
\draw    (423.16,92.48) .. controls (422.78,105.85) and (407.44,91.86) .. (405.22,110.19) ;
\draw [shift={(405,113)}, rotate = 272.44] [fill={rgb, 255:red, 0; green, 0; blue, 0 }  ][line width=0.08]  [draw opacity=0] (8.04,-3.86) -- (0,0) -- (8.04,3.86) -- (5.34,0) -- cycle    ;
\draw    (431.8,229.4) -- (431.8,294.48) ;
\draw [shift={(431.8,297.48)}, rotate = 270] [fill={rgb, 255:red, 0; green, 0; blue, 0 }  ][line width=0.08]  [draw opacity=0] (8.04,-3.86) -- (0,0) -- (8.04,3.86) -- (5.34,0) -- cycle    ;
\draw    (233.96,125.88) .. controls (259.24,125.78) and (242.25,159.75) .. (270.05,162.15) ;
\draw [shift={(272.76,162.28)}, rotate = 180.73] [fill={rgb, 255:red, 0; green, 0; blue, 0 }  ][line width=0.08]  [draw opacity=0] (8.04,-3.86) -- (0,0) -- (8.04,3.86) -- (5.34,0) -- cycle    ;
\draw    (314.92,54.68) -- (314.92,56.94) -- (314.92,73.28) ;
\draw [shift={(314.92,76.28)}, rotate = 270] [fill={rgb, 255:red, 0; green, 0; blue, 0 }  ][line width=0.08]  [draw opacity=0] (8.04,-3.86) -- (0,0) -- (8.04,3.86) -- (5.34,0) -- cycle    ;
\draw    (350.92,54.68) -- (350.92,73.28) ;
\draw [shift={(350.92,76.28)}, rotate = 270] [fill={rgb, 255:red, 0; green, 0; blue, 0 }  ][line width=0.08]  [draw opacity=0] (8.04,-3.86) -- (0,0) -- (8.04,3.86) -- (5.34,0) -- cycle    ;
\draw    (422.92,54.68) -- (422.92,73.28) ;
\draw [shift={(422.92,76.28)}, rotate = 270] [fill={rgb, 255:red, 0; green, 0; blue, 0 }  ][line width=0.08]  [draw opacity=0] (8.04,-3.86) -- (0,0) -- (8.04,3.86) -- (5.34,0) -- cycle    ;
\draw    (29.4,173.4) -- (54.96,173.4) ;
\draw [shift={(57.96,173.4)}, rotate = 180] [fill={rgb, 255:red, 0; green, 0; blue, 0 }  ][line width=0.08]  [draw opacity=0] (8.04,-3.86) -- (0,0) -- (8.04,3.86) -- (5.34,0) -- cycle    ;

\draw (62,166.4) node [anchor=north west][inner sep=0.75pt]    {$\vec{x}_{\circ }$};
\draw (114.4,115.2) node [anchor=north west][inner sep=0.75pt]    {$G[\svc ,\!\vec{0}]$};
\draw (115.4,206.2) node [anchor=north west][inner sep=0.75pt]    {$G[\svc ,\!\vec{1}]$};
\draw (275,157) node [anchor=north west][inner sep=0.75pt]  [font=\scriptsize] [align=center] {Mix-and-Match\\(\Cref{alg:mix-and-match})};
\draw (302.13,207.67) node [anchor=north west][inner sep=0.75pt]  [font=\normalsize]  {$d^{( 1)}$};
\draw (330.47,238.33) node [anchor=north west][inner sep=0.75pt]  [font=\normalsize]  {$\hat{p}^{( 1)}$};
\draw (307.87,297.87) node [anchor=north west][inner sep=0.75pt]  [font=\normalsize]  {$\hat{p}^{( 1)}$};
\draw (343,297.4) node [anchor=north west][inner sep=0.75pt]  [font=\normalsize]  {$\hat{p}^{( 2)}$};
\draw (420.87,297.87) node [anchor=north west][inner sep=0.75pt]  [font=\normalsize]  {$\hat{p}^{(q)}$};
\draw (381.37,300.97) node [anchor=north west][inner sep=0.75pt]  [font=\normalsize]  {$\cdots $};
\draw (307,34.4) node [anchor=north west][inner sep=0.75pt]  [font=\normalsize]  {$\vec{a}^{( 1)}$};
\draw (343,34.4) node [anchor=north west][inner sep=0.75pt]  [font=\normalsize]  {$\vec{a}^{( 2)}$};
\draw (417,34.4) node [anchor=north west][inner sep=0.75pt]  [font=\normalsize]  {$\vec{a}^{( q)}$};
\draw (378,41.4) node [anchor=north west][inner sep=0.75pt]  [font=\normalsize]  {$\cdots $};
\draw (390.07,180.46) node [anchor=north west][inner sep=0.75pt]  [font=\Large,rotate=-330]  {$\cdots $};
\draw (47.56,66.5) node [anchor=north west][inner sep=0.75pt]  [color={rgb, 255:red, 0; green, 0; blue, 0 }  ,opacity=1 ] [align=left] {\Cref{alg:inn-prod-reduc-new}};
\draw (204.6,180) node [anchor=north west][inner sep=0.75pt]    {$\mathcal{R}_{1}$};
\draw (203.8,198.8) node [anchor=north west][inner sep=0.75pt]    {$\mathcal{R}_{2}$};
\draw (196.8,237.2) node [anchor=north west][inner sep=0.75pt]    {$\mathcal{R}_{2kn}$};
\draw (212.6,212.6) node [anchor=north west][inner sep=0.75pt]    {$\vdots $};
\draw (172.6,173) node [anchor=north west][inner sep=0.75pt]    {$\vec{\mathcal{R}}$};
\draw (204.6,85) node [anchor=north west][inner sep=0.75pt]    {$\mathcal{R}_{1}$};
\draw (203.8,103.8) node [anchor=north west][inner sep=0.75pt]    {$\mathcal{R}_{2}$};
\draw (196.8,142.2) node [anchor=north west][inner sep=0.75pt]    {$\mathcal{R}_{2kn}$};
\draw (212.6,117.6) node [anchor=north west][inner sep=0.75pt]    {$\vdots $};
\draw (172.6,81) node [anchor=north west][inner sep=0.75pt]    {$\vec{\mathcal{R}}$};
\draw (307.4,72.2) node [anchor=north west][inner sep=0.75pt]  [font=\normalsize]  {$\vec{a}_{\circ }^{( 1)}$};
\draw (343.4,72.2) node [anchor=north west][inner sep=0.75pt]  [font=\normalsize]  {$\vec{a}_{\circ }^{( 2)}$};
\draw (417.4,72.2) node [anchor=north west][inner sep=0.75pt]  [font=\normalsize]  {$\vec{a}_{\circ }^{( q)}$};
\draw (378.4,79.2) node [anchor=north west][inner sep=0.75pt]  [font=\normalsize]  {$\cdots $};
\draw (16,166.8) node [anchor=north west][inner sep=0.75pt]    {$\vec{x}$};

\end{tikzpicture}

\caption{An outline of the reduction of \Cref{lem:inn-prod-reduction-new}. The secret vector $\sv$ is accessed by the randomizers run on the query graphs $G[\svc, \vec 0]$ and $G[\svc, \vec 1]$; the remaining computations are post-processing.
}
\label{fig:reduction-outline}

\end{figure}
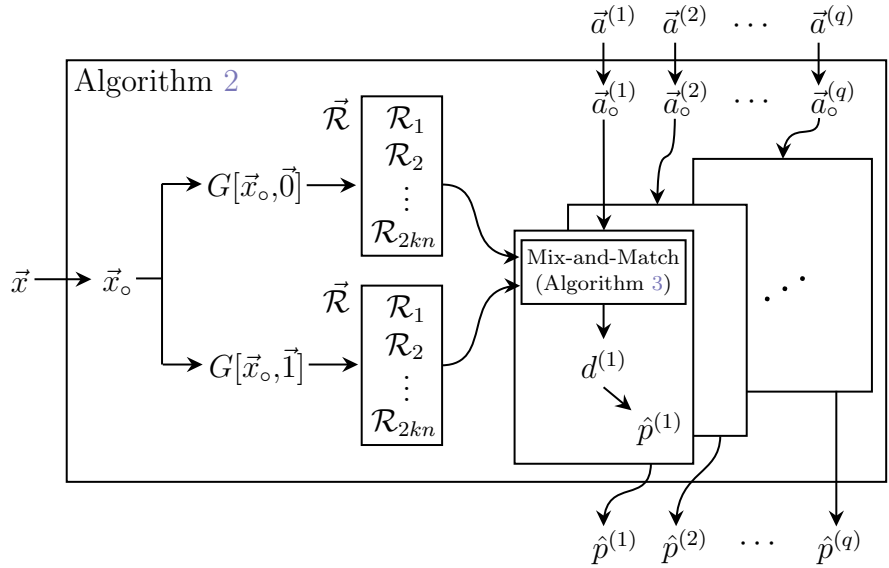

%% file: algs/alg-reduction-polynomial.tex
\begin{algorithm}[ht!]
    \caption{Reduction from inner-product queries to $H$-counting.}
    \label{alg:inn-prod-reduc-new}
    
    \begin{algorithmic}[1]
        \Statex \textbf{Parameters:}
        Dimension $n\in\N$, number of queries $q\in\N$,
        graph $H = ([k], E(H))$, $H$-counting algorithm $\cA = (\vec \cR, \cP)$ for graphs on node set $[2kn]$ using randomizers $\vec{\cR}_\pubrand = (\cR_{v, \pubrand})_{v \in [2kn]}$, public randomness distribution $\distpubrand$,
        and post-processing algorithm $\cP$.
        \Statex \textbf{Input:} Dataset $\sv \in \{0, 1\}^n$, queries $\iq^{(1)}, \ldots, \iq^{(q)} \in \{0, 1\}^n$.
        
        \Statex \textbf{Output:} Estimates $\wh p^{(1)}, \ldots, \wh p^{(q)}$ of $\inn{\sv, \iq^{(1)}}, \ldots, \inn{\sv, \iq^{(q)}}$, respectively.

        \smallskip

        \State Draw and fix $\pubrand\sim \distpubrand$.
        \State Let $\svc \gets \sv \circ (\vec 1 - \sv)$, and $\iqc^{(\ell)} \gets \iq^{(\ell)} \circ \vec 0$ for each $\ell \in [q]$.
        \Comment{\commentstyle{$\|\svc\| = n, \|\iqc^{(\ell)}\| = \|\iq^{(\ell)}\|$ are public.}}
        
        \For{{\bf each} $c\in \set{0,1}$}
            \State $\left(\nodeview[v, c]\right)_{v \in [2kn]} \gets \vec \cR_\pubrand \left(G[\svc, %
            c\cdot \vec 1_{2n}]\right)$.
            \Comment{\commentstyle{Query secret dataset $\sv$.}}
            \label{line:construct-vec-r-new}
        \EndFor
        \State Let $\nodeview$ 
        denote the vector of the $4kn$ outputs produced on Line~\ref{line:construct-vec-r-new}.
        \For{{\bf each} $\ell \in [q]$}
            \State $d^{(\ell)} \gets 
            \mixandmatch \left(2n, H, \cA, \pubrand,
            \iqc^{(\ell)}, \nodeview \right)$.
            \Comment{\commentstyle{$\mixandmatch$ is defined in Alg.~\ref{alg:mix-and-match}.}}
            \State Define $g_\ell(\, \cdot\, ) \gets \f(\, \cdot \, , n, \|\iq^{(\ell)}\|)$.
            \label{line:def-g}
            \State 
            $\wh p^{(\ell)} \gets \arg\min_{p \in \{0, \ldots, \|\iq^{(\ell)}\|\}} \left|g_\ell(p) - d^{(\ell)}\right|$. %
            \label{line:def-p}
        \EndFor
    \State \textbf{return} $\wh p^{(1)}, \ldots, \wh p^{(q)}$.
    \end{algorithmic}
\end{algorithm}

%% file: algs/mix-and-match-polynomial.tex
\begin{algorithm}[ht!]
    \caption{\mixandmatch}
    \label{alg:mix-and-match}
    
    \begin{algorithmic}[1]
        \Statex \textbf{Parameters:} number of nodes $n \in \mathbb{N}$, graph $H = ([k], E(H))$, $H$-counting algorithm $\cA = (\vec \cR, \cP)$ for graphs on node set $[kn]$ using randomizers $\vec{\cR} = (\cR_{i,\rho})_{i \in [kn]}$, post-processing algorithm $\cP$, public randomness $\pubrand$.
        \Statex \textbf{Input:} Query $\iq = (a_1, \ldots, a_n) \in \{0, 1\}^n$, vector $\nodeview$ from randomizers run on $G[\sv, \iq]$.
        
        \Statex \textbf{Output:} Estimate $d$ of $\subgraphcount(G[\sv, \iq])$.

        \smallskip

        \For{{\bf each} $i \in [n]$}
            \For{{\bf each} $s \in \{1, 2\}$}
                \State $r(\node{s}{i}) \gets 
                \nodeview[\node{s}{i}, a_i 
                ]
                $.
                \label{line:lb-set-ui-new}
            \EndFor
            
            \For{{\bf each} $s \in \{3, \ldots, k\}$}
                \State $r(\node{s}{i}) \gets \cR_{\node{s}{i}, \rho} \bparen{N_{\node{s}{i}} (G[\vec 0, \, \iq])}
                $.
                \label{line:lb-set-w-new}
            \EndFor
        \EndFor

        \State \textbf{return} $d \gets \cP \left( \left(r(v)\right)_{v \in [kn]}\right)$.
        \Comment{\commentstyle{$d \sim \cA(G[\sv, \iq])$ by %
        \Cref{lem:mix-and-match}.
        }}
        \label{line:lb-inn-answer-new}
        
    \end{algorithmic}
\end{algorithm}